\documentclass[11pt,letterpaper]{article}

\usepackage[T1]{fontenc}
\usepackage{fullpage}
\usepackage{mathpazo}
\usepackage{amsmath,amssymb,amsthm,amsfonts}
\usepackage{array}
\usepackage{booktabs}
\usepackage{float}
\usepackage[dvipsnames]{xcolor}
\usepackage{microtype}
\usepackage{tabularx,colortbl}
\usepackage[style=alphabetic,natbib=true,maxalphanames=9,minalphanames=3,maxcitenames=9,mincitenames=3,maxbibnames=99]{biblatex}
\usepackage{enumitem}
\usepackage{listings}
\usepackage[colorlinks=true,allcolors=magenta]{hyperref}
\usepackage[nameinlink,capitalize,noabbrev]{cleveref}
\makeatletter
\AddToHook{cmd/appendix/before}{\def\cref@section@alias{appendix}\def\cref@subsection@alias{appendix}}
\makeatother
\usepackage{mathtools}
\usepackage{pifont}
\usepackage{thm-restate}
\usepackage{nicefrac}
\usepackage{tikz}
\usepackage{pgfplots}
\pgfplotsset{compat=1.17}
\usetikzlibrary{arrows.meta}
\usepackage{todonotes}

\AtBeginBibliography{\sloppy}

\usepackage[ruled]{algorithm2e} 
\SetAlFnt{\small}
\SetAlCapFnt{\small}
\SetAlCapNameFnt{\small}
\SetAlCapHSkip{0pt}
\IncMargin{-\parindent}
\crefname{algocf}{algorithm}{algorithms}
\Crefname{algocf}{Algorithm}{Algorithms}

\newtheorem{theorem}{Theorem}
\newtheorem{proposition}{Proposition}

\newtheorem{lemma}{Lemma}
\newtheorem{open}{Open Question}
\newcounter{definition}
\theoremstyle{definition}

\theoremstyle{remark}

\usepackage{tcolorbox}

\newcommand{\definitionboxtitle}[1]{Definition~\thedefinition\if\relax\detokenize{#1}\relax\else: #1\fi}
\newenvironment{definition}[1][]{\refstepcounter{definition}\begin{tcolorbox}[
    colback=MidnightBlue!4,
    colframe=MidnightBlue!70!black,
    colbacktitle=MidnightBlue!70!black,
    coltitle=white,
    fonttitle=\normalfont\small\bfseries,
    title={\definitionboxtitle{#1}},
    rounded corners,
    left=6pt,
    right=6pt,
    top=3pt,
    bottom=3pt,
    toptitle=3pt,
    bottomtitle=3pt,
    before skip=8pt,
    after skip=8pt,
    before upper={
  \setlength{\abovedisplayskip}{4pt}\setlength{\belowdisplayskip}{4pt}\setlength{\parskip}{4pt}}
  ]}{\end{tcolorbox}}
\crefname{definition}{definition}{definitions}
\Crefname{definition}{Definition}{Definitions}
\newtcolorbox{programbox}[1]{
  colback=gray!4,
  colframe=black,
  colbacktitle=black,
  coltitle=white,
  fonttitle=\normalfont\small\scshape,
  title={#1},
  rounded corners,
  left=6pt,
  right=6pt,
  top=2pt,
  bottom=2pt,
  toptitle=3pt,
  bottomtitle=3pt,
  before skip=8pt,
  after skip=8pt,
  before upper={
  \setlength{\abovedisplayskip}{4pt}\setlength{\belowdisplayskip}{4pt}}
}
\renewcommand{\ge}{\geqslant}
\renewcommand{\geq}{\geqslant}
\renewcommand{\le}{\leqslant}
\renewcommand{\leq}{\leqslant}

\newcommand{\set}[1]{\left\{#1\right\}}
\DeclarePairedDelimiter{\setsize}{|}{|}

\DeclareMathOperator{\E}{\mathbb{E}}

\newcommand{\diff}{\mathop{}\!\mathrm{d}}
\newcommand{\R}{\mathbb{R}}

\numberwithin{equation}{section}
\allowdisplaybreaks
\newcommand{\SC}{\operatorname{SC}}
\newcommand{\dist}{\operatorname{dist}}
\newcommand{\supp}{\operatorname{supp}}

\newcommand{\SL}[1]{\operatorname{SL}_{#1}}
\newcommand{\RSL}[1]{\operatorname{RSL}_{#1}}

\newcommand{\cav}{\operatorname{cav}}

\newcommand{\loss}{\mathcal{L}}

\title{Improving Randomized Metric Distortion to $2.1441$\thanks{A previous version of the manuscript provided a weaker distortion bound of $2.3282$. That version used an infinite-dimensional linear program and its dual to bound distortion. The present version replaces that approach with a tighter potential-based analysis. A fresh grid search points to the draw distribution for which a distortion bound of $2.1441$ can be verified exactly. This version also proves a lower bound of $2.136$ for every rule in the family of rules we consider. Hence, for the particular rule used in our upper bound, the potential method's guarantee is less than $0.009$ above the rule's exact distortion.}}
\author{Nisarg Shah\\University of Toronto\\\texttt{nisarg@cs.toronto.edu}}
\makeatletter
\let\@date\@empty
\makeatother

\hypersetup{
  pdftitle={Improving Randomized Metric Distortion to 2.1441},
  pdfauthor={Nisarg Shah}
}

\begin{document}

\maketitle

\begin{abstract}
In metric social choice, each voter ranks a set of $m$ candidates by her distance to them in an unknown metric space. The cost of a candidate is its average distance to the voters. A randomized voting rule must use only the rankings to choose a lottery over candidates. Its distortion is the worst-case ratio between the expected cost under the lottery it returns and the cost of the best candidate. \citet{CRWW24} prove an upper bound of $2.753$, establishing a constant separation from deterministic rules, for which the best achievable distortion is $3$. Recently, \citet{Fra26} and \citet{Ye26} independently improve the bound to $2.5$.

We break this barrier by defining, proving the existence of, and using a new family of rules: \emph{random-size stable lotteries}. Let $D$ be a random variable over the domain of positive integers. A random-size stable lottery $\RSL{D}$ guarantees that the probability of a random voter preferring any fixed candidate $c$ to her favorite of $D$ i.i.d. draws from $\RSL{D}$ is at most $\E[1/(D+1)]$, where the probability averages over the voter, the value of $D$, and the $D$ draws from $\RSL{D}$. When $D=1$ deterministically, this is precisely the well-known maximal lottery. More generally, when $D=k$ deterministically, this reduces to its stable $k$-lottery generalization, which \citet{CRTW25} prove the existence of for every fixed positive integer $k$. Their minimax argument for a fixed $k$ easily generalizes to a random $D$. 

Our main contribution is to show how stability with respect to a random $D$ can be used to bound distortion directly. Specifically, we show that, for an explicit $D$ supported on $1$, $2$, $9$, and $10$, any random-size stable lottery $\RSL{D}$ has distortion at most $2.1441$. The proof combines the biased-metric characterization of distortion due to \citet{CRWW24} with a new potential argument, and exactly verifies a rational distortion bound via nonnegativity testing of polynomials in the Bernstein basis. We also prove a lower bound of $2.136$ on the distortion of any randomized voting rule that, for any fixed $D$, always outputs a random-size stable lottery $\RSL{D}$. This is strictly higher than the best known universal lower bound of approximately $2.1126$ on all randomized voting rules~\citep{CR22}.

All the proofs have been obtained using GPT-5.6-Sol with significant guidance from the author and verified by the author, who significantly expanded on the exposition and simplified arguments with the aid of GPT-5.6-Sol and Claude Opus 5.
\end{abstract}

\newpage
\setcounter{tocdepth}{2}
\tableofcontents
\medskip
\newpage

\section{Introduction}

Social choice theory studies how to choose one of several alternative decisions (candidates) by aggregating the conflicting preferences of multiple parties (voters). The norm is to elicit ranked preferences from voters over the candidates because exact utilities or costs are often difficult for voters to identify and report. Consequently, much of the centuries-old literature has sought voting rules that satisfy qualitative axioms defined on this ordinal data. However, no single axiom has been convincing on its own and seeking multiple axioms together has often led to impossibility results rather than a choice of the ``optimal voting rule''~\citep{Arr51,Gib73,Sat75}. The rise of approximation algorithms in computer science has offered a new perspective: even when voters report ordinal preferences, one can aim to optimize welfare defined on the underlying numerical utilities or costs, with success measured by \emph{distortion}---the approximation ratio between the achieved welfare and the optimum in the worst case over the unknown preference intensities. This has reignited the search for optimal voting rules, well-defined in this distortion framework.

Specifically, in the metric social choice model~\citep{ABEPS18}, $n$ voters and $m$ candidates are points in the same underlying metric space. A voter's distance to a candidate is her cost for that candidate, and she ranks candidates from smallest to largest distance. The social cost of a candidate is its average distance to the voters. A deterministic (resp., randomized) voting rule must choose a candidate (resp., a lottery over candidates) based only on the rankings, and its \emph{distortion} is the ratio of the (expected) social cost of its choice to the smallest social cost of any candidate, in the worst case over all metrics consistent with the observed rankings.

\citet{GHS20} prove that the best distortion achievable by deterministic rules is $3$. To see why deterministic rules cannot do any better and why randomization \emph{may} help, consider a simple example with two voters ($1,2$) and two candidates ($a,b$). Voter $1$ prefers $a$ to $b$ while voter $2$ prefers $b$ to $a$. Let the underlying metric have distance $d$ and $\Delta := d(a,b)$. Only two consistent possibilities are (provably) of interest:
\begin{itemize}
    \item either voter $1$ is located at $a$ and voter $2$ is located at the midpoint of $a$ and $b$, which makes the social costs of $a$ and $b$ equal to $\Delta/4$ and $3\Delta/4$, respectively,
    \item or voter $1$ is located at the midpoint and voter $2$ is located at $b$, which makes the social costs of $a$ and $b$ equal to $3\Delta/4$ and $\Delta/4$, respectively.
\end{itemize}

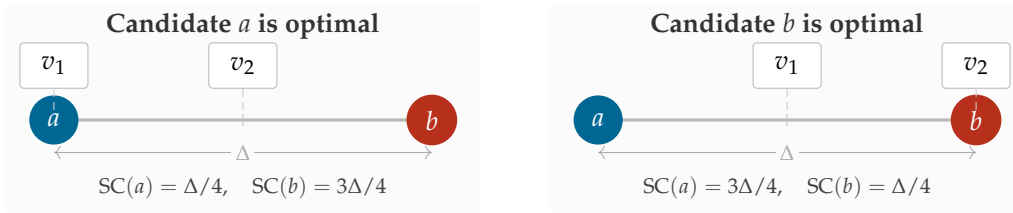
\begin{figure}[H]
\centering
\begin{tikzpicture}[
  candidate/.style={circle,draw=none,minimum size=6.5mm,font=\bfseries\small,text=white},
  voter/.style={rounded corners=1.5pt,draw=gray!55,fill=white,minimum width=9mm,minimum height=6mm,font=\small},
  metric/.style={line width=1.2pt,draw=gray!55},
  guide/.style={densely dashed,draw=gray!40},
]
\foreach \shift/\title/\vone/\vtwo/\sca/\scb in {
  0/{Candidate $a$ is optimal}/0/2.5/{\Delta/4}/{3\Delta/4},
  7.2/{Candidate $b$ is optimal}/2.5/5/{3\Delta/4}/{\Delta/4}
}{
  \begin{scope}[xshift=\shift cm]
  \fill[gray!3,rounded corners=3pt] (-0.65,-1.25) rectangle (5.65,1.55);
  \node[font=\small\bfseries,text=black!80] at (2.5,1.25) {\title};
  \draw[metric] (0,0) -- (5,0);
  \draw[gray!45] (0,-0.08) -- (0,0.08) (2.5,-0.08) -- (2.5,0.08) (5,-0.08) -- (5,0.08);
  \node[candidate,fill=MidnightBlue] at (0,0) {$a$};
  \node[candidate,fill=BrickRed] at (5,0) {$b$};
  \node[voter] (v1) at (\vone,0.72) {$v_1$};
  \node[voter] (v2) at (\vtwo,0.72) {$v_2$};
  \draw[guide] (v1.south) -- (\vone,0.12);
  \draw[guide] (v2.south) -- (\vtwo,0.12);
  \draw[<->,gray!65] (0,-0.43) -- node[fill=gray!3,inner sep=1.5pt,font=\scriptsize] {$\Delta$} (5,-0.43);
  \node[font=\scriptsize,text=black!75] at (2.5,-0.89) {$\SC(a)=\sca,\quad \SC(b)=\scb$};
  \end{scope}
}
\end{tikzpicture}
\caption{Two metrics consistent with the reported rankings, which show a lower bound of $3$ for deterministic rules and $2$ for randomized rules.}
\label{fig:det-rand-gap}
\end{figure}

Thus, as \Cref{fig:det-rand-gap} makes visible, either candidate can be optimal while the other is $3$ times worse, making $3$ a barrier that deterministic rules cannot cross. However, choosing between $a$ and $b$ uniformly at random achieves an expected social cost of $\Delta/2$, which is only $2$ times the minimum possible cost and that is the best possible distortion on this instance. It was even conjectured that $2$ may be the optimal metric distortion of randomized rules, but \citet{PS21,CR22} disprove it by establishing a lower bound strictly greater than $2$, with the bound by the latter converging to approximately $2.1126$ as $m \to \infty$.

On the upper-bound side, \citet{CRWW24} prove a distortion guarantee of $2.753$ by combining \emph{maximal lottery}, the Nash equilibrium of a natural zero-sum game dating back to the 1960s~\citep{Kre65,Fis84}, and RaDiUS, a novel rule they design. Their work, which received the SODA 2024 Best Paper Award, establishes that randomized rules can in fact achieve a constant distortion strictly lower than the best guarantee of $3$ for deterministic rules. Recently, \citet{Fra26} and \citet{Ye26} independently improve this to $2.5$ with a rule of striking simplicity: the equal mixture of a maximal lottery and a novel rule termed \emph{Integrated Veto} (also, \emph{Veto Lottery}), a randomized variant of the optimal deterministic rule \emph{Plurality Veto}~\citep{KK22}.\footnote{Plurality Veto is a refinement of the Plurality Matching rule of \citet{GHS20}, both achieving the optimal metric distortion of $3$ for deterministic rules.} This still leaves a fundamental question open:

\begin{quote}
\emph{What is the smallest distortion that a randomized voting rule can achieve using only voters' preference rankings?}
\end{quote}

\subsection{Our Results}\label{sec:our-results}

The main result of our work improves the upper bound, closing about $92\%$ of the gap between the upper bound of $2.5$ and the lower bound of $2.1126$ left by prior work.

\begin{restatable}{theorem}{mainupperbound}\label{thm:intro-main-upper}
There exists a randomized voting rule with metric distortion at most
\[
\frac{107200742653}{50000000000}=2.14401485306.
\]
\end{restatable}

To accomplish this, we consider a generalization of maximal lotteries to stable $k$-lotteries ($\SL{k}$) by \citet{CRTW25}, who use them to design so-called tournament voting rules with low metric distortion; stable $1$-lotteries are precisely maximal lotteries. We design a further generalization of this to \emph{random-size stable lotteries}: given a \emph{draw distribution} $D$ over positive integers, its corresponding random-size stable lottery $\RSL{D}$ effectively provides the guarantee of $\SL{k}$ but with $k$ drawn from $D$. Yet, we show that $\RSL{D}$ cannot in general be obtained by simply mixing $\SL{k}$ for different values of $k$; we prove its existence for any $D$ directly via a minimax argument similar to what \citet{CRTW25} used for $\SL{k}$, and give a multiplicative-weights approximation algorithm with an explicit running time when $D$ has finite support. The main technical contribution of our work is to connect the guarantee provided by $\RSL{D}$ to metric distortion through a novel potential function and an elegant telescoping argument applied to the biased-metric characterization of distortion by \citet{CRWW24}. The final distortion expression lends itself to a computer-assisted search for a suitable draw distribution $D^\star$. This is followed by an exact verification using rational arithmetic, which reduces the problem to nonnegativity testing of polynomials using the so-called Bernstein basis. 

While the upper bound of \Cref{thm:intro-main-upper} is eerily close to the best-known universal lower bound of approximately $2.1126$ by \citet{CR22}, our next result shows that the random-size stable lottery rule corresponding to any draw distribution $D$, which always returns an $\RSL{D}$ lottery, has distortion more than $2.136$, proving that this remaining gap cannot be closed simply by choosing a different draw distribution $D$. Note that the lower bound holds only when $D$ is fixed independently of the profile.

\begin{restatable}{theorem}{integerlowerbound}\label{thm:intro-integer-lower}
For any positive-integer-valued random variable $D$, there exists a profile on which every random-size stable lottery $\RSL{D}$ has metric distortion strictly greater than
\[
\frac{267}{125}=2.136.
\]
Hence, this lower bound applies to the metric distortion of any randomized voting rule that always returns a random-size stable lottery $\RSL{D}$ corresponding to any fixed $D$. 
\end{restatable}

Together, \Cref{thm:intro-main-upper,thm:intro-integer-lower} place the best distortion achievable by the random-size stable lottery rule corresponding to any draw distribution $D$ in the interval $[2.136,2.14401485306]$. \Cref{tab:bounds-landscape} shows our upper and lower bounds in the context of prior work.

\begin{table}[htb!]
\centering
\begin{tabularx}{\textwidth}{@{}rXl@{}}
\toprule
\multicolumn{3}{@{}l@{}}{\emph{Upper bounds}}\\
\midrule
$3$ & optimal deterministic rule & \citep{GHS20}\\
$2.753$ & a mixture of maximal lottery and RaDiUS & \citep{CRWW24}\\
$2.5$ & the equal mixture of maximal lottery and Integrated Veto & \citep{Fra26,Ye26}\\
\rowcolor{lightgray} $2.1441$ & a random-size stable lottery rule & this paper\\
\midrule
\multicolumn{3}{@{}l@{}}{\emph{Lower bounds}}\\
\midrule
$2.0631$ & an example with $m=7$ & \citep{PS21}\\
$2.1126$ & a construction for any $m \ge 3$, bound on the left as $m \to \infty$ & \citep{CR22}\\
\rowcolor{lightgray} $2.136$ & any random-size stable lottery rule & this paper\\
\bottomrule
\end{tabularx}
\caption{Randomized metric distortion bounds from our work and prior work.}
\label{tab:bounds-landscape}
\end{table}

\subsection{Technical Overview}\label{sec:technical-overview}

\paragraph{Random-size stable lotteries.}
Our new ingredient starts from stable $k$-lotteries ($\SL{k}$)~\citep{CRTW25}, which generalize maximal lotteries and satisfy
\[
\Pr\nolimits_v[c\succ_v(\SL{k})^k]\leq\frac1{k+1}
\qquad\text{for every }c\in C.
\]
Here $(\SL{k})^k$ consists of $k$ independent draws from $\SL{k}$, $v$ is a uniformly random voter, and the probability also includes the draws and the prescribed tie-breaking. We let the number of draws itself be random. For a positive-integer-valued random variable $D$, a random-size stable lottery $\RSL{D}$ satisfies
\[
\Pr\nolimits_v[c\succ_v(\RSL{D})^D]\leq\E\left[\frac1{D+1}\right]
\qquad\text{for every }c\in C.
\]
When $D=k$ deterministically, this recovers $\SL{k}$. Crucially, one lottery must be chosen before $D$ is realized: independently choosing an $\SL{k}$ for each $k$ and then mixing those lotteries does not work because repeated draws from the mixture create uncontrolled cross terms. A convex minimax argument nevertheless proves that $\RSL{D}$ exists for every $D$. The definition and proof appear in \Cref{sec:comparison}; an approximation algorithm is deferred to \Cref{app:unification}.

\paragraph{Biased-metric characterization.} The exact biased-metric characterization of \citet{CRWW24}, adapted from an earlier linear programming approach of \citet{CR22}, replaces the worst case over all metrics consistent with the observed ranked preference profile by a vector of $m$ nonnegative offsets, with the offset of a designated optimal candidate normalized to zero. For each voter, the largest decrease between the offsets of a more-preferred and a less-preferred candidate determines her contribution to the optimal candidate's cost, while the smallest offset among a candidate and all candidates ranked below it determines her contribution to that candidate's excess cost over the optimum. Bounding distortion therefore reduces to comparing the lottery's aggregate excess cost with the aggregate largest decrease for every offset vector.

\paragraph{Why a random number of draws helps.}
Fix a voter and a reference candidate $c$, and suppose that $\RSL{D}$ places mass $t$ on candidates strictly below $c$ and no mass on $c$ itself. The probability that $c$ beats all $D$ draws is
\[
\phi(t):=\E[t^D]=\sum_{d\geq1}\Pr[D=d]\,t^d.
\]
It is tempting to use the same formula without the zero-mass assumption, but this is incorrect because draws equal to the reference candidate create ties. Under the random-label tie-breaking defined in \Cref{sec:comparison}, if $b$ is the mass strictly below the reference candidate and $z$ is the mass on it, then conditioning on its label $s\in[0,1]$ gives the corrected probability $\int_0^1\phi(b+sz)\diff s$, as formalized in \Cref{lem:compare-c-PD}. Averaging this over the voters gives a value at most $\E[1/(D+1)]$, with equality when $\RSL{D}$ puts a positive probability mass on $c$. The draw distribution determines the shape of $\phi$, and therefore how strongly stability constrains different positions in the voters' rankings.

Coming back to the offset formulation via biased metrics, for each value $t$, consider the candidates selected with positive probability whose offsets exceed $t$. These are the only ones that can contribute to the lottery's excess cost at that value of $t$. The tightness established above limits how poorly these candidates can be positioned across the voters' rankings on average, and thereby bounds their contribution to the excess cost. We express this bound as the difference of a one-dimensional potential $W$ at two ranking percentiles. As $t$ increases, the sets of candidates with offsets above $t$ are nested, and the endpoint used at one level becomes the starting point at the next. The potential differences therefore telescope, leaving only $W(1)$. The biased-metric characterization then gives distortion at most $1+2W(1)$. Randomizing $D$ lets us shape $\phi$ so that a single potential provides the required bound for every $t$. Computer-assisted search finds a $D^\star$ supported on $\set{1,2,9,10}$, for which we construct a potential $W$ satisfying the required inequalities with $W(1)=57200742653/100000000000$. The resulting distortion bound of $107200742653/50000000000$ can be verified via exact rational arithmetic.

\paragraph{A lower bound for every integer draw distribution.} We complement the potential upper bound with two explicit constructions of hard preference profiles. Each construction turns one scalar inequality involving the probability-generating function of $D$ into an actual distortion witness. We use one instance of the first construction and two instances of the second. Except when $D=1$ deterministically (handled separately), strict convexity makes the random-size stable lottery unique, so the construction applies to every way of selecting a stable lottery; a compactness argument converts weighted voters into a finite unweighted election. Finally, we show that no choice of the draw distribution can make all three resulting instances have distortion at most $267/125$.

\subsection{Related Work}\label{sec:related-work}

\paragraph{Metric distortion.} The metric model was introduced by \citet{ABEPS18}, who proved a lower bound of $3$ for every deterministic rule and an upper bound of $5$ for the Copeland rule. \citet{GKM17} showed that Ranked Pairs can have distortion at least $5$ and, more broadly, that even randomization cannot take tournament rules below distortion $3$. \citet{MW19} reduced the upper bound to $2+\sqrt{5}\approx 4.236$ using generalized uncovered sets and proposed a candidate-selection condition whose universal existence would imply distortion $3$; \citet{Kem20LP} subsequently developed an LP-duality framework that unified these analyses and led to closely related sufficient conditions. \citet{GHS20} finally matched the lower bound with Plurality Matching, settling the optimal deterministic distortion at $3$; \citet{KK22,KK23} subsequently obtained the same optimal guarantee with Plurality Veto and Simultaneous Plurality Veto. For randomized rules, \citet{CRWW24} combine maximal lottery and a rule they call RaDiUS to obtain distortion $2.753$. \citet{Fra26} and \citet{Ye26} independently improve the bound to $2.5$ using the equal mixture of maximal lottery and Integrated Veto, called Veto Lottery by \citet{Ye26}. Our random-size stable lottery improves this to $2.1441$.

\paragraph{Utilitarian distortion.} In the utilitarian model initiated by \citet{PR06}, voters again report only rankings, but the hidden cardinal values are otherwise unrestricted and the objective is to maximize their sum. Here normalization is essential: rankings are unchanged if one rescales each voter's utilities independently, so rank data alone provide no way to place different voters' utilities on a common scale. The standard unit-sum model normalizes each voter's utilities to sum to one, giving every voter equal weight while retaining the relative intensity of her preferences. Metric distortion needs no analogous normalization because a common underlying metric already places all voters' costs on the same scale. Under unit-sum utilities, the optimal deterministic distortion is $\Theta(m^2)$ \citep{CP11,CNP+17}; \citet{BCHL+15} proved an $\Omega(\sqrt{m})$ lower bound for randomized rules, and \citet{EKPS24} matched it with an $O(\sqrt{m})$ rule, closing the randomized case asymptotically as well. Although the metric and utilitarian models are each motivated by the same loss of cardinal information, rules optimized for one can perform poorly in the other. \citet{GLS23} nevertheless design deterministic and randomized rules that simultaneously achieve near-optimal guarantees in both models, establishing best-of-both-worlds distortion.

\paragraph{Information--distortion tradeoffs.} Full rankings are only one point on a broader elicitation spectrum. At one end, top-$t$ ballots ask each voter to rank only her $t$ favorite candidates; work under both utilitarian and metric objectives quantifies what can still be guaranteed as this reported prefix shrinks \citep{Kem20Comm,BHLS22}. At the other end, a value query asks a voter for her exact utility for one candidate in addition to her ranking. A few such queries can sharply reduce utilitarian distortion \citep{ABFV21}; for every fixed number of value queries per voter, \citet{ES25} achieve distortion polynomial in $\min\set{n,m}$, with exponent equal to the reciprocal of the number of queries, matching the known dependence on $m$ when $n=\Omega(m)$. More generally, one can optimize the elicitation format itself subject only to a bit budget. \citet{MPSW19,MSW20} characterize this communication--distortion frontier up to logarithmic factors: achieving utilitarian distortion $d$ requires $\widetilde{\Theta}(m/d)$ bits per voter with deterministic elicitation and $\widetilde{\Theta}(m/d^3)$ bits with randomized elicitation. Hence, the format of the elicited information can matter as much as its quantity: a carefully designed short message may yield lower distortion than a substantially longer ranked ballot.

\section{Preliminaries}\label{sec:model}

Let $C$ be a finite set of $m$ candidates and $V$ a finite set of $n$ voters. For a positive integer $r$, write $[r]=\set{1,\ldots,r}$. Each voter $v\in V$ reports a strict ranking $\succ_v$ of $C$; write $i\succcurlyeq_v j$ if $i=j$ or $i\succ_v j$. The \emph{profile} is the collection $\sigma=(\succ_v)_{v\in V}$. Unless stated otherwise, $v$ is drawn uniformly at random from $V$ in every expectation or probability carrying the subscript $v$.

A \emph{metric instance} places the voters and candidates in a common metric space with distance $d$.\footnote{Distinct points are allowed to be at distance zero, so technically this is a pseudometric.} It is \emph{consistent} with $\sigma$ if every voter's ranking agrees with her distances, that is, $a\succ_vb$ implies $d(v,a)\leq d(v,b)$ for all $v\in V$ and $a,b\in C$. The social cost of a candidate $c$ and the expected social cost of a lottery $P\in\Delta(C)$, the set of probability distributions on $C$, are
\[
\SC_d(c)=\E_v[d(v,c)]
\qquad\text{and}\qquad
\SC_d(P)=\E_{c\sim P}[\SC_d(c)].
\]

\begin{definition}[Distortion]
For a profile $\sigma$ and a lottery $P$, define
\[
\dist_\sigma(P)=\sup_{d\text{ consistent with }\sigma}
\frac{\SC_d(P)}{\min_{c\in C}\SC_d(c)},
\]
with $0/0=1$ and $z/0=+\infty$ for $z>0$. We drop the subscript when the profile is fixed. 

A randomized rule $\mathcal R$ maps each profile $\sigma$ to a lottery, and its \emph{distortion} is $\sup_\sigma\dist_\sigma(\mathcal R(\sigma))$. Note that the rule is unaware of the underlying metric.
\end{definition}

\section{Random-Size Stable Lotteries}\label{sec:comparison}

We now formally define our key novel ingredient: random-size stable lotteries. We prove their existence by extending the minimax argument of \citet{CRTW25} from a fixed number of draws $k$ to a random number of draws $D$. In \Cref{app:unification}, we show that this also yields a standard multiplicative-weights algorithm to compute such lotteries approximately and establish other properties of interest that we will not use in our work.

\paragraph{Comparing multisets.} First, we need to lift each voter's comparison of individual candidates to a comparison between multisets of candidates. Fix a voter $v$ and two nonempty multisets $S$ and $T$. It is natural to compare $S$ and $T$ by the candidates that $v$ ranks highest in them. If these candidates are distinct, $v$'s ranking determines the comparison; the only ambiguity arises when the same candidate is highest-ranked in both multisets, possibly with different multiplicities. To resolve this ambiguity, we use the random-label tie-breaking scheme of \citet{CRTW25}: give every duplicate an independent label drawn uniformly from $[0,1]$, and let $v$ compare duplicates by their labels. Thus, if voter $v$'s highest-ranked candidate in $S\cup T$ has multiplicity $r$ in $S$ and $s$ in $T$, then $\Pr[S\succ_v T]=r/(r+s)$. By choosing $S$ or $T$ to be a singleton, this also yields comparisons between individual candidates and multisets.

\paragraph{Random sample size.} Let $D$ be a positive-integer-valued random variable indicating the number of draws. For a lottery $P$ and a positive integer $d$, $P^d$ and $P^D$ denote the multisets obtained by $d$ and $D$ i.i.d.\ draws from $P$, respectively. To analyze the comparison between a candidate $c$ and the multiset $P^D$, it will be helpful to introduce the following well-known quantity, which is an equivalent representation of a random variable.

\begin{definition}[Probability Generating Function]\label{def:probability-generating-function}
    The probability generating function $\phi$ of $D$ is 
    \begin{equation*}
    \phi(t)=\E[t^D]=\sum_{d\geq1}\Pr[D=d] \cdot t^d.
    \end{equation*}
    For any $D$, $\phi$ is increasing and convex on $[0,1]$, with $\phi(0)=0$, $\phi(1)=1$, and $\int_0^1\phi(t)\diff t=\E[1/(D+1)]$. Unless $D=1$ deterministically, $\phi$ is strictly convex.
\end{definition}

The following lemma uses this function to compare a candidate $c$ to the multiset $P^D$.

\begin{lemma}\label{lem:compare-c-PD}
Fix any voter $v$, lottery $P$, and candidate $c$. Then,
\[
\Pr[c \succ_v P^D] = \int_0^1 \phi(b+sz) \diff s,
\]
where $b = P(\set{j : c \succ_v j})$ is the $P$-mass on candidates $v$ ranks strictly below $c$ and $z = P(c)$. When $z>0$, the same probability is $\frac{1}{z}\int_b^{b+z}\phi(u)\diff u$. This probability is continuous and convex in $P$.
\end{lemma}
\begin{proof}
Let us first analyze the probability that $c$ beats a random draw from $P$. Let $s \in [0,1]$ be the random label assigned to $c$. Then, $c$ beats all candidates ranked strictly below it (probability mass $b$) and an additional $s$ fraction of its own probability mass which would be assigned a label smaller than $s$ (probability mass $s \cdot z$). Hence, the total probability of $c$ defeating a single draw from $P$ is $b+sz$. The probability of it defeating $d$ i.i.d.\ draws from $P$ is then $(b+sz)^d$. Averaging over the number of draws $d$ sampled from $D$ and the random label $s$ uniform in $[0,1]$, we get
\begin{align*}
\Pr[c \succ_v P^D] = \E_{d \sim D} \Pr[c \succ_v P^d] &= \E_{d \sim D} \int_0^1 (b+sz)^d \diff s \\
&= \int_0^1 \left(\E_{d \sim D} (b+sz)^d\right) \diff s = \int_0^1 \phi(b+sz) \diff s,
\end{align*}
where the penultimate step uses Fubini's theorem to exchange the integral and expectation. When $z>0$, the alternative expression follows by the variable substitution $u:=b+sz$.

For each fixed $d$ and $s$, the map $P\mapsto(b+sz)^d$ is continuous and convex because $b+sz$ is linear in $P$ and $u\mapsto u^d$ is continuous and convex. Averaging over $d$ and $s$ preserves convexity. For continuity, if $P_k\to P$, the corresponding integrands converge pointwise and lie in $[0,1]$, so dominated convergence over $d$ and $s$ applies.
\end{proof}

\begin{definition}[Random-size stable lottery]\label{def:random-size-stable}
We say that a lottery is \emph{random-size stable} for $D$, and denote such a lottery by $\RSL{D}$, if
\begin{equation}\label{eq:comparison-system}
\forall c \in C: \Pr\nolimits_v[c\succ_v(\RSL{D})^D]\leq\E\left[\frac1{D+1}\right].
\end{equation}
\end{definition}

If we have a challenger lottery $A$ instead of a fixed challenger $c$, we can also equivalently write, denoting $A = A^1$ to be a random draw from $A$,
\[
\Pr\nolimits_v[A\succ_v(\RSL{D})^D]\leq\E\left[\frac1{D+1}\right].
\]
This is at least as strong as the condition in \Cref{def:random-size-stable}, and equivalence follows from noticing that the left hand side is linear in $A$ and, hence, maximized at some fixed $c$.

When $D=1$ deterministically, \Cref{def:random-size-stable} is precisely a maximal lottery; when $D=k$ deterministically, it is a stable $k$-lottery defined by \citet{CRTW25}, who prove its existence using a minimax argument.\footnote{\citet{CJMW20,JMW20} define a related notion of stable lotteries over committees, which are more restrictive in that the committee does not necessarily come from independent draws from the same lottery, but more general in other dimensions, as discussed in \Cref{app:unification}.} The next result shows that the same argument extends to a random number of draws $D$, proving the existence of $\RSL{D}$ for all $D$.

\begin{theorem}[Existence of random-size stable lotteries]\label{thm:comparison-existence}
For every finite preference profile and every positive-integer-valued random variable $D$, a random-size stable lottery $\RSL{D}$ exists. Further, recall \Cref{eq:comparison-system} defining $\RSL{D}$:
\begin{equation*}
\forall c \in C: \Pr\nolimits_v[c\succ_v(\RSL{D})^D]\leq\E\left[\frac1{D+1}\right].
\end{equation*}
No smaller right-hand side works for every profile, and every candidate in the support of $\RSL{D}$ satisfies this with equality.
\end{theorem}

\begin{proof}
Let us abbreviate $\mu=\E[1/(D+1)]$. For a challenger lottery $A$ and a lottery $P$, define
\[
F(A,P)=\Pr\nolimits_v[A\succ_v P^D].
\]
The function is linear in $A$. \Cref{lem:compare-c-PD} establishes its continuity and convexity in $P$. Von Neumann's minimax theorem therefore gives
\[
\min_P\max_A F(A,P)=\max_A\min_P F(A,P).
\]
For every lottery $P$, exchangeability among $d+1$ independent labeled draws gives $\Pr\nolimits_v[P\succ_v P^d]=1/(d+1)$, and hence $F(P,P)=\mu$. Choosing $P=A$ shows that the right-hand side of the minimax equality is at most $\mu$, while $\max_A F(A,P)\geq F(P,P)=\mu$ for every $P$. The minimax value is therefore exactly $\mu$. Because $F$ is linear in $A$, an optimal $P$, denoted $\RSL{D}$, satisfies \Cref{eq:comparison-system} for every candidate. Finally,
\[
\mu=F(\RSL{D},\RSL{D})
=\sum_{c\in C}\RSL{D}(c)\,\Pr\nolimits_v[c\succ_v(\RSL{D})^D].
\]
Every term in this weighted average is at most $\mu$, so equality holds for every candidate in the support. The same identity shows that no smaller universal right-hand side is possible.
\end{proof}

Random-size stable lotteries cannot in general be obtained by simply averaging stable $k$-lotteries for different values of $k$; \Cref{app:mixing-counterexample} gives a two-candidate example.

The use of the minimax theorem also yields a standard multiplicative-weights algorithm for approximating $\RSL{D}$; the precise guarantee appears as \Cref{prop:comparison-computation} in \Cref{app:unification}.

An interesting identity, in which we reverse which side receives the $D$ draws, is given as \Cref{lem:reverse-identity} in \Cref{app:unification}.

\section{Bounding the Metric Distortion via a Potential Function}\label{sec:potential-analysis}

We now connect random-size stability directly to metric distortion. The argument uses the exact biased-metric characterization of \citet{CRWW24}, adapted from \citet{CR22}, which reduces the worst case over all consistent metrics to worst case over $m$-dimensional vectors of nonnegative offsets.

\subsection{Biased-metric characterization} 
\begin{definition}[Offset Loss and Offset Scale]\label{def:offset-loss}
Given an offset vector $x\in\mathbb R_{\geq0}^C$, define
\[
s_{v,c}(x)=\min_{j:\,c\: \succcurlyeq_v\: j} x_j,
\qquad
\Delta_v(x)=\max_{i,j : i\: \succcurlyeq_v\: j}(x_i-x_j).
\]

For a lottery $P$ and a voter $v$, let the \emph{offset loss of voter $v$} and the (overall) \emph{offset loss}, respectively, be
\[
\loss_v(P,x)=\sum_{c\in C}P(c)\,s_{v,c}(x),
\qquad
\loss(P,x)=\E_v[\loss_v(P,x)].
\]
The \emph{offset scale} is $\E_v[\Delta_v(x)]$. 
\end{definition}

In the metric supplied by the characterization below, the offset loss is the excess social cost of $P$ over the designated optimal candidate, while the offset scale is twice that candidate's social cost.

\begin{theorem}[Reduction to offsets~\citep{CR22,CRWW24}]\label{thm:biased-metrics}
Fix a profile and a lottery $P$. Then
\[
\dist(P)=\sup_{\substack{x\in\R_{\geq0}^C:\:\min_c x_c=0}}\left(1+2\frac{\loss(P,x)}{\E_v[\Delta_v(x)]}\right),
\]
with $0/0=0$ and $z/0=+\infty$ for $z>0$ in the displayed fraction. Further, for each candidate $o$ and offset vector $x \in \R_{\ge 0}^C$ with $x_o = 0$, there exists a metric $d$ consistent with the profile such that
\[
\SC_d(o)=\frac12\E_v[\Delta_v(x)]
\qquad\text{and}\qquad
\SC_d(c)-\SC_d(o)=\E_v[s_{v,c}(x)], \forall c, 
\]
under which $o$ is socially optimal and $\SC_d(P)/\SC_d(o) = 1+2\loss(P,x)/\E_v[\Delta_v(x)]$.
\end{theorem}

For a finite profile, or more generally any distribution over rankings, we call the expression inside the supremum in \Cref{thm:biased-metrics} the \emph{offset value} of $(P,x)$. On a finite profile, for every $\lambda\geq0$, the lottery $P$ has distortion at most $1+2\lambda$ if and only if $\loss(P,x)\leq\lambda\E_v[\Delta_v(x)]$ for every nonnegative offset vector $x$ with $\min_{c\in C}x_c=0$.

\subsection{Ranking coordinates and win probabilities}

We next develop a consequence of random-size stability that will let us bound $\loss(P,x)$. Recall the probability generating function $\phi$ of $D$ from \Cref{def:probability-generating-function}. We now define its integral and total mass.

\begin{definition}[Integral and Stability Benchmark]\label{def:potential-draw-functions}
For a positive-integer-valued random variable $D$ with probability generating function $\phi$, define
\[
\Phi(u)=\int_0^u\phi(s)\diff s,
\qquad
\mu=\Phi(1)=\E\left[\frac1{D+1}\right].
\]
\end{definition}

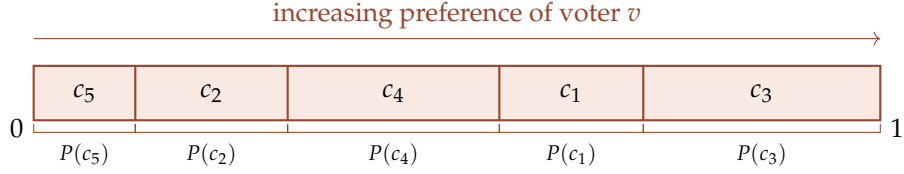
\begin{figure}[t]
\centering
\begin{tikzpicture}[x=11.2cm,y=1cm]
\fill[BrickRed!8] (0,0) rectangle (1,0.72);
\draw[BrickRed!75!black,line width=0.8pt] (0,0) rectangle (1,0.72);
\foreach \x in {0.12,0.30,0.55,0.72}{\draw[BrickRed!75!black,line width=0.8pt] (\x,0) -- (\x,0.72);}
\node[font=\small] at (0.06,0.36) {$c_5$};
\node[font=\small] at (0.21,0.36) {$c_2$};
\node[font=\small] at (0.425,0.36) {$c_4$};
\node[font=\small] at (0.635,0.36) {$c_1$};
\node[font=\small] at (0.86,0.36) {$c_3$};
\draw[->,BrickRed!75!black] (0,1.08) -- node[above,font=\small] {increasing preference of voter $v$} (1,1.08);
\node[anchor=east,font=\small] at (0,-0.10) {$0$};
\node[anchor=west,font=\small] at (1,-0.10) {$1$};
\draw[BrickRed!75!black] (0,-0.05) -- (0,-0.16) -- (0.12,-0.16) -- (0.12,-0.05);
\draw[BrickRed!75!black] (0.12,-0.05) -- (0.12,-0.16) -- (0.30,-0.16) -- (0.30,-0.05);
\draw[BrickRed!75!black] (0.30,-0.05) -- (0.30,-0.16) -- (0.55,-0.16) -- (0.55,-0.05);
\draw[BrickRed!75!black] (0.55,-0.05) -- (0.55,-0.16) -- (0.72,-0.16) -- (0.72,-0.05);
\draw[BrickRed!75!black] (0.72,-0.05) -- (0.72,-0.16) -- (1,-0.16) -- (1,-0.05);
\node[anchor=north,font=\scriptsize] at (0.06,-0.17) {$P(c_5)$};
\node[anchor=north,font=\scriptsize] at (0.21,-0.17) {$P(c_2)$};
\node[anchor=north,font=\scriptsize] at (0.425,-0.17) {$P(c_4)$};
\node[anchor=north,font=\scriptsize] at (0.635,-0.17) {$P(c_1)$};
\node[anchor=north,font=\scriptsize] at (0.86,-0.17) {$P(c_3)$};
\end{tikzpicture}
\caption{The $P$-mass arranged from voter $v$'s least- to most-preferred candidate.}
\label{fig:p-mass-ranking}
\end{figure}

Next, we propose a visual interpretation that will allow us to understand and use \Cref{lem:compare-c-PD}, which expresses $\Pr[c \succ_v P^D]$ in terms of $\phi$, at a much more granular level. 

\begin{definition}[Rank intervals, win probabilities, and deviations]\label{def:potential-rank-mass}
For a lottery $P$ and a voter $v$, arrange the mass placed by $P$ on the different candidates on a line, from the voter's least-preferred candidate at coordinate $0$ to her most-preferred candidate at coordinate $1$, as illustrated in \Cref{fig:p-mass-ranking}.

Each candidate $c$ occupies a block $I_{v,c}$ of length $P(c)$. For $H\subseteq\supp(P)$, define
\[
H_v=\bigcup_{c\in H}I_{v,c},
\qquad
M_v(H)=\int_{H_v}\phi(u)\diff u,
\qquad
\widetilde{M}_v(H)=M_v(H)-P(H)\mu.
\]

Thus, $H_v$ is the part of voter $v$'s interval occupied by $H$. A labeled copy of candidate $c$ corresponds to a uniformly random point $u \in I_{v,c}$. By \Cref{lem:compare-c-PD}, this copy beats $D$ i.i.d.\ labeled draws from $P$ with probability $\phi(u)$. Hence, $M_v(H)$ is the probability that a challenger drawn from $P$ lies in $H$ and beats $D$ i.i.d.\ draws from $P$ in voter $v$'s ranking. We call $M_v(H)$ the \emph{win probability} of $H$ for voter $v$ and $\widetilde{M}_v(H)$ its \emph{win-probability deviation} from the benchmark $P(H)\mu$.
\end{definition}

The next lemma justifies calling $P(H) \mu$ the benchmark. It is indeed the average of $M_v(H)$, making $\widetilde{M}_v(H)$ a mean-zero deviation. 

\begin{lemma}\label{lem:potential-moment}
Let $P$ be random-size stable for $D$. Then, for every $H\subseteq\supp(P)$, we have $\E_v[\widetilde{M}_v(H)]=0$.
\end{lemma}
\begin{proof}
Fix a voter $v$ and candidate $c \in \supp(P)$. \Cref{lem:compare-c-PD} shows that
\[
\Pr[c \succ_v P^D] = \frac{1}{P(c)}\int_{I_{v,c}}\phi(u)\diff u.
\]
The final statement of \Cref{thm:comparison-existence} says that this probability, averaged over voters, is exactly $\mu$ for every $c\in\supp(P)$. Multiplying by $P(c)$ and summing over $c\in H$ gives $\E_v[M_v(H)]=P(H)\mu$, which is equivalent to the claim.
\end{proof}

For any voter $v$, $P(H)\mu$ is the natural benchmark for $M_v(H)$: $P(H)$ is the probability that the challenger belongs to $H$, while $\mu$ is the probability that a distinguished draw is best among $D+1$ i.i.d.\ draws. Thus, $\widetilde{M}_v(H)$ is positive when challengers from $H$ win more often than this benchmark from voter $v$'s perspective, and negative when they win less often. \Cref{lem:potential-moment} says that these deviations cancel when averaged over voters.

\subsection{Bounding the offset loss}

\Cref{thm:biased-metrics} shows that the distortion of any lottery $P$ can be upper bounded by $1+2\lambda$ if we show $\loss(P,x) = \E_v[\loss_v(P,x)] \le \lambda \cdot \E_v[\Delta_v(x)]$ for all offset vectors $x \in \R_{\ge 0}^C$ with $x_o = 0$ for at least one candidate $o$. Fix any designated optimal candidate $o$ and nonnegative offset vector $x$ with $x_o = 0$. Further, fix any threshold $t \ge 0$ and mark each candidate $c$ with $x_c > t$. 

\paragraph{Proof strategy.} If we had a voter-level bound $\loss_v(P,x) \le \lambda \cdot \Delta_v(x)$ for each voter $v$, then the above inequality would follow immediately by averaging over $v$. However, this may not hold for a reasonable value of $\lambda$. Instead, we will shift each voter's loss by adding $\widetilde{M}_v(H)$ with the same multiplier for every voter, and then bound the shifted loss by $\lambda \cdot \Delta_v(x)$. Because the shifts are mean-zero (\Cref{lem:potential-moment}), they cancel out when averaging over $v$, which would still imply the desired inequality.

A final piece of the puzzle is to choose $H$. For a fixed threshold $t$, we will let $H$ contain the marked candidates in the support of $P$ whose offsets exceed $t$; integrating over all possible $t$ in the end will account for the entire offset.

Another objective we will achieve along the way is to compress all useful information from the profile into three quantities, a lower cutoff $a_v(t)$, a marked probability $p(t)$, and an upper cutoff $A_v(t)$, as functions of $t$. This will be helpful in computer-assisted search later on. We will begin by expressing $\loss_v(P,x)$ in terms of the lower cutoff $a_v(t)$, and later express the correction term involving $M_v(H)$ in terms of all three quantities.

\begin{definition}[Lower cutoff]\label{def:threshold-voter-loss}
For each voter $v$, consider the $P$-mass arranged from her least-preferred candidate to her most-preferred in \Cref{fig:p-mass-ranking}. Scan this interval from $0$ to $1$ and stop upon reaching the first unmarked candidate. The \emph{lower cutoff} $a_v(t)$ is the total probability under $P$ of the candidates passed before the scan stops.
\end{definition}

In \Cref{fig:p-mass-ranking}, this scan certifies that the interval from $0$ through $a_v(t)$ is occupied by marked candidates. Its connection to the offset loss is exact.

\begin{lemma}[Decomposing the offset loss into lower cutoffs]\label{lem:potential-threshold-decomposition}
For every voter $v$ and threshold $t\geq0$,
\[
a_v(t)=\Pr\nolimits_{c\sim P}[s_{v,c}(x)>t].
\]
Consequently,
\begin{align*}
\loss_v(P,x)&=\int_0^\infty a_v(t)\diff t,\\
\loss(P,x)&=\int_0^\infty\E_v[a_v(t)]\diff t.
\end{align*}
\end{lemma}
\begin{proof}
Recall that $s_{v,c}(x)=\min_{j:\,c\succcurlyeq_v j}x_j$. Thus, $s_{v,c}(x)>t$ if and only if $c$ and every candidate ranked below $c$ by $v$ are marked, which holds exactly when the scan passes $c$. This proves the formal expression for $a_v(t)$. The tail-integral identity for the nonnegative random variable $s_{v,c}(x)$ then gives
\[
\loss_v(P,x)=\E_{c\sim P}[s_{v,c}(x)]=\int_0^\infty\Pr\nolimits_{c\sim P}[s_{v,c}(x)>t]\diff t=\int_0^\infty a_v(t)\diff t.
\]
Averaging over voters proves the second equality.
\end{proof}

\begin{definition}[Marked probability and upper cutoff]\label{def:potential-cutoffs}
Let $H(t)=\set{c\in\supp(P):x_c>t}$ be the marked support set, and let $p(t)=P(H(t))$ be the total probability that $P$ places on marked candidates. For voter $v$, let
\[
H_v(t)=\bigcup_{c\in H(t)}I_{v,c}
\]
be the part of her ranking interval occupied by these candidates. An \emph{upper cutoff} $A_v(t)$ is any coordinate satisfying $H_v(t)\subseteq[0,A_v(t)]$.
\end{definition}

The scan in \Cref{def:threshold-voter-loss} passes only marked candidates, so $[0,a_v(t)]\subseteq H_v(t)$. The marked region $H_v(t)$ has length $p(t)$ and, by the definition of the upper cutoff, lies inside $[0,A_v(t)]$. Hence,
\[
[0,a_v(t)]\subseteq H_v(t)\subseteq[0,A_v(t)]
\qquad\text{and}\qquad
0\leq a_v(t)\leq p(t)\leq A_v(t)\leq1.
\]
The smallest upper cutoff ends at the block of voter $v$'s most-preferred marked candidate. We also allow larger choices because this flexibility will be helpful when comparing successive thresholds. \Cref{fig:threshold-ranking} illustrates the two cutoffs and the marked probability.

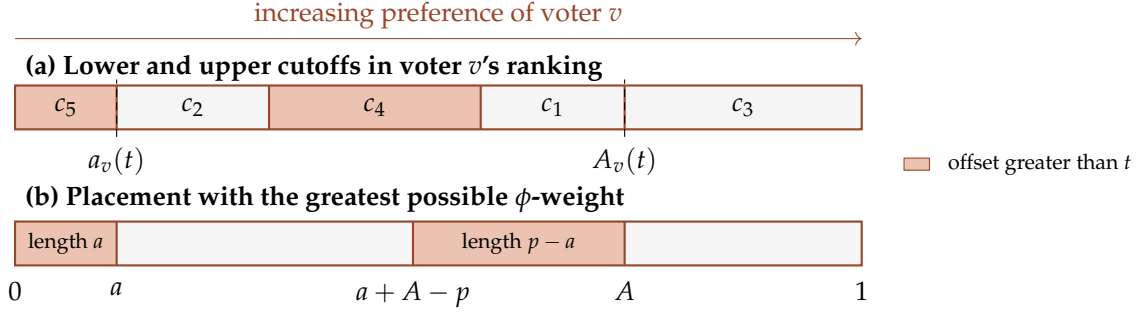
\begin{figure}[H]
\centering
\begin{tikzpicture}[x=11.2cm,y=1cm]
\draw[->,BrickRed!75!black] (0,3.30) -- node[above,font=\small] {increasing preference of voter $v$} (1,3.30);
\node[anchor=west,font=\small\bfseries] at (0,2.92) {(a) Lower and upper cutoffs in voter $v$'s ranking};
\fill[black!4] (0,2.10) rectangle (1,2.68);
\fill[BrickRed!24] (0,2.10) rectangle (0.12,2.68);
\fill[BrickRed!24] (0.30,2.10) rectangle (0.55,2.68);
\draw[BrickRed!75!black,line width=0.8pt] (0,2.10) rectangle (1,2.68);
\foreach \z in {0.12,0.30,0.55,0.72}{\draw[BrickRed!75!black,line width=0.8pt] (\z,2.10) -- (\z,2.68);}
\node[font=\small] at (0.06,2.39) {$c_5$};
\node[font=\small] at (0.21,2.39) {$c_2$};
\node[font=\small] at (0.425,2.39) {$c_4$};
\node[font=\small] at (0.635,2.39) {$c_1$};
\node[font=\small] at (0.86,2.39) {$c_3$};
\draw[densely dashed] (0.12,2.00) -- (0.12,2.76);
\draw[densely dashed] (0.72,2.00) -- (0.72,2.76);
\node[anchor=north,font=\small] at (0.12,2.00) {$a_v(t)$};
\node[anchor=north,font=\small] at (0.72,2.00) {$A_v(t)$};
\fill[BrickRed!24] (1.05,1.56) rectangle (1.08,1.72);
\draw[BrickRed!75!black] (1.05,1.56) rectangle (1.08,1.72);
\node[anchor=west,font=\scriptsize] at (1.09,1.64) {offset greater than $t$};
\node[anchor=west,font=\small\bfseries] at (0,1.18) {(b) Placement with the greatest possible $\phi$-weight};
\fill[black!4] (0,0.30) rectangle (1,0.88);
\fill[BrickRed!24] (0,0.30) rectangle (0.12,0.88);
\fill[BrickRed!24] (0.47,0.30) rectangle (0.72,0.88);
\draw[BrickRed!75!black,line width=0.8pt] (0,0.30) rectangle (1,0.88);
\foreach \z in {0.12,0.47,0.72}{\draw[BrickRed!75!black,line width=0.8pt] (\z,0.30) -- (\z,0.88);}
\node[font=\scriptsize] at (0.06,0.59) {length $a$};
\node[font=\scriptsize] at (0.595,0.59) {length $p-a$};
\node[anchor=north,font=\small] at (0,0.22) {$0$};
\node[anchor=north,font=\small] at (0.12,0.22) {$a$};
\node[anchor=north,font=\small] at (0.47,0.22) {$a+A-p$};
\node[anchor=north,font=\small] at (0.72,0.22) {$A$};
\node[anchor=north,font=\small] at (1,0.22) {$1$};
\end{tikzpicture}
\caption{The example continues \Cref{fig:p-mass-ranking}. At the threshold shown, $H(t)=\set{c_5,c_4}$, so the marked probability is $p(t)=P(c_5)+P(c_4)$. The scan for voter $v$ passes $c_5$ and stops at the unmarked candidate $c_2$, giving the lower cutoff $a_v(t)=P(c_5)$. The displayed upper cutoff $A_v(t)$ is deliberately nonminimal: it lies to the right of both marked blocks. Panel (b) keeps the forced interval $[0,a]$ and moves the remaining marked length $p-a$ as far right as the upper cutoff $A$ permits.}
\label{fig:threshold-ranking}
\end{figure}

The next lemma allows us to forget the exact locations of the marked blocks by bounding the win probabilities $M_v(H(t))$ in terms of only three quantities: the lower cutoff $a_v(t)$, their total length $p(t)$, and the upper cutoff $A_v(t)$.

\begin{lemma}[Bounding the win probabilities in $(a,p,A)$]\label{lem:potential-moment-rearrangement}
Let $0\leq a\leq p\leq A\leq1$, and let $S \subseteq [0,1]$ be a measurable set of length $p$ such that $[0,a] \subseteq S \subseteq [0,A]$. Then,
\[
\int_S\phi(u)\diff u\leq\Phi(a)+\Phi(A)-\Phi(a+A-p).
\]
In particular, for every voter $v$ and threshold $t$,
\begin{equation}\label{eq:potential-moment-rearrangement}
M_v(H(t))\leq\Phi(a_v(t))+\Phi(A_v(t))-\Phi(a_v(t)+A_v(t)-p(t)).
\end{equation}
\end{lemma}

\begin{proof}
The part of $S$ outside $[0,a]$ has length $p-a$. Because $\phi$ is nondecreasing, its integral is largest when this remaining length is placed at the right end of $[0,A]$, namely on $[a+A-p,A]$, as in \Cref{fig:threshold-ranking}(b). This gives
\[
\int_S\phi(u)\diff u\leq\int_0^a\phi(u)\diff u+\int_{a+A-p}^A\phi(u)\diff u=\Phi(a)+\Phi(A)-\Phi(a+A-p).
\]
Applying this bound to $S=H_v(t)$, using $[0,a_v(t)] \subseteq H_v(t) \subseteq [0,A_v(t)]$, yields \Cref{eq:potential-moment-rearrangement}.
\end{proof}

\paragraph{Adding a mean-zero correction.}
The lower cutoff $a_v(t)$ is the contribution that we need to bound, while \Cref{lem:potential-moment-rearrangement} bounds the win probability $M_v(H(t))$. The win-probability deviation
\[
\widetilde{M}_v(H(t))=M_v(H(t))-p(t)\mu
\]
links the two: because $H(t)$ and $p(t)$ do not depend on the voter, \Cref{lem:potential-moment} gives $\E_v[\widetilde{M}_v(H(t))]=0$. We may therefore add any voter-independent multiple of this deviation to $a_v(t)$ without changing its average.

Suppose $0<p(t)<1$. Then $p(t)\mu-\Phi(p(t))>0$ because the average of the nondecreasing, nonconstant function $\phi$ on $[0,p(t)]$ is smaller than its average on $[0,1]$. Consider the reference placement in which the scan passes every candidate in the marked support set. Here $a_v(t)=p(t)$, $H_v(t)=[0,p(t)]$, and $M_v(H(t))=\Phi(p(t))$. We choose the multiplier so that the corrected lower cutoff is zero in this reference placement.

\begin{definition}[Lower-cutoff correction]\label{def:potential-corrected-loss}
When $0<p(t)<1$, define the \emph{lower-cutoff correction} and the \emph{corrected lower cutoff}, respectively, by
\begin{equation}\label{eq:potential-corrected-loss}
\begin{aligned}
\operatorname{corr}_v(t)&=\frac{p(t)}{p(t)\mu-\Phi(p(t))}\,\widetilde{M}_v(H(t)),\\
Q_v(t)&=a_v(t)+\operatorname{corr}_v(t).
\end{aligned}
\end{equation}
\end{definition}

The multiplier is positive. Thus, substituting the win-probability bound \Cref{eq:potential-moment-rearrangement} into \Cref{eq:potential-corrected-loss} bounds $Q_v(t)$ using only the three scalars retained from the ranking. We name this scalar bound next.

\begin{definition}[Scalar bounds for corrected lower cutoffs]\label{def:potential-corrected-bounds}
For $0<p<1$ and $0\leq a\leq p\leq A\leq1$, define
\begin{equation}\label{eq:potential-edge}
\mathcal E_\phi(a,p,A)=a+\frac{p\bigl(\Phi(a)+\Phi(A)-\Phi(a+A-p)-p\mu\bigr)}{p\mu-\Phi(p)},
\end{equation}
and set $\mathcal E_\phi(0,0,A)=0$. For $r\in[0,1]$, define
\[
\mathcal G_\phi(r)=r-\frac{\phi(r)-\mu}{1-\mu}.
\]
\end{definition}

The function $\mathcal E_\phi$ handles $p(t)<1$. When $p(t)=1$, the denominator in \Cref{eq:potential-corrected-loss} vanishes, and we use $\mathcal G_\phi$ instead. Indeed, $x_o=0\leq t$ makes $o$ unmarked, and $p(t)=1$ therefore implies $P(o)=0$. If $r_v$ is the probability assigned to candidates that voter $v$ ranks below $o$, then the lower-cutoff scan stops no later than $o$, giving $a_v(t)\leq r_v$, while \Cref{lem:compare-c-PD} shows that stability against $o$ bounds the average of $\phi(r_v)$. The following lemma records both cases.

\begin{lemma}[Correcting one lower cutoff]\label{lem:potential-local-correction}
Let $P$ be random-size stable for $D$. Fix a candidate $o$, a nonnegative offset vector $x$ with $x_o=0$, and a threshold $t\geq0$. Use the lower cutoff from \Cref{def:threshold-voter-loss} and the quantities from \Cref{def:potential-cutoffs}. If $p(t)=0$, set $Q_v(t)=0$. If $0<p(t)<1$, let $Q_v(t)$ be the corrected lower cutoff from \Cref{def:potential-corrected-loss}. Then, for every upper cutoff $A_v(t)$,
\[
\E_v[Q_v(t)]=\E_v[a_v(t)]
\qquad\text{and}\qquad
Q_v(t)\leq\mathcal E_\phi(a_v(t),p(t),A_v(t)).
\]
If $p(t)=1$, let $r_v$ be the probability assigned to candidates that voter $v$ ranks below $o$ and set $Q_v(t)=\mathcal G_\phi(r_v)$. Then $\E_v[Q_v(t)]\geq\E_v[a_v(t)]$.
\end{lemma}

\begin{proof}
Abbreviate $p=p(t)$. If $p=0$, then $a_v(t)=0$ for every voter, so both claims are immediate. If $0<p<1$, \Cref{lem:potential-moment} gives $\E_v[\operatorname{corr}_v(t)]=0$, proving the equality of averages. The pointwise bound follows by substituting \Cref{eq:potential-moment-rearrangement} into \Cref{eq:potential-corrected-loss}.

Suppose $p=1$. As observed above, $P(o)=0$ and $a_v(t)\leq r_v$. By \Cref{lem:compare-c-PD}, stability against $o$ gives $\E_v[\phi(r_v)]\leq\mu$. Therefore,
\[
\E_v[Q_v(t)]=\E_v[r_v]+\frac{\mu-\E_v[\phi(r_v)]}{1-\mu}\geq\E_v[r_v]\geq\E_v[a_v(t)].
\]
\end{proof}

\paragraph{The one-threshold reduction.} At a fixed threshold, the exact positions of the marked candidates affect $M_v(H(t))$, but \Cref{lem:potential-local-correction} removes this dependence. For the remaining argument, we retain only the lower cutoff $a_v(t)$, the marked probability $p(t)$, and an upper cutoff $A_v(t)$; $\mathcal E_\phi(a_v(t),p(t),A_v(t))$ bounds the corrected lower cutoff.

\paragraph{From one threshold to all thresholds.}
To combine the one-threshold bounds for a fixed voter, we compare thresholds separated by $\Delta_v(x)$. At each threshold after the first, the definition of $\Delta_v(x)$ ensures that every candidate still marked is ranked below the candidate that stopped the scan at the preceding threshold. Consequently, the preceding lower cutoff is an upper cutoff at the current threshold.

\paragraph{The remaining analytic goal.} We seek a potential $W$ such that $\mathcal E_\phi(a,p,A)\leq W(A)-W(a)$ and $\mathcal G_\phi(r)\leq W(1)-W(r)$. When the preceding lower cutoff is used as the current upper cutoff, these potential differences telescope across thresholds and leave only $W(1)$.

The next lemma completes the reduction from the profile-dependent metric problem to a one-dimensional potential. Its conclusion has the factor $1+2W(1)$ because the offset characterization identifies the optimal candidate's cost with one half of the offset scale.

\begin{lemma}[Potential criterion]\label{lem:rsl-potential}
Let $P$ be random-size stable for $D$. Suppose there is a nondecreasing function $W:[0,1]\to\mathbb R$ with $W(0)=0$ such that
\begin{align}
\mathcal E_\phi(a,p,A)&\leq W(A)-W(a) &&\text{for every }0\leq a\leq p\leq A\leq1\text{ with }p<1,\label{eq:potential-internal-condition}\\
\mathcal G_\phi(r)&\leq W(1)-W(r) &&\text{for every }r\in[0,1].\label{eq:potential-boundary-condition}
\end{align}
Then $P$ has metric distortion at most $1+2W(1)$.
\end{lemma}

\begin{proof}
Fix a nonnegative offset vector $(x_c)_{c\in C}$ and a designated candidate $o$ with $x_o=0$. For each voter $v$, abbreviate $\Delta_v=\Delta_v(x)$ as in \Cref{thm:biased-metrics}. For every threshold $t\geq0$, let $(Q_v(t))_v$ be the corrected lower cutoffs supplied by \Cref{lem:potential-local-correction}.

We first prove the deterministic inequality
\begin{equation}\label{eq:potential-telescoping}
\int_0^\infty Q_v(t)\diff t\leq W(1)\Delta_v.
\end{equation}
If $\Delta_v=0$, the offsets cannot increase as we move upward through voter $v$'s ranking. At every threshold, the scan therefore passes all candidates in the marked support set before it stops, so $H_v(t)=[0,p(t)]$ and $a_v(t)=p(t)$. Taking $A_v(t)=p(t)$ bounds every corrected lower cutoff with $p(t)<1$ by $\mathcal E_\phi(p(t),p(t),p(t))=0$. If $p(t)=1$, every supported candidate lies below $o$, so $r_v=1$ and $\mathcal G_\phi(1)=0$. This proves \Cref{eq:potential-telescoping} when $\Delta_v=0$.

Suppose $\Delta_v>0$. Every threshold $t\geq0$ can be written uniquely as $t=s+k\Delta_v$, where $s\in[0,\Delta_v)$ and $k$ is a nonnegative integer. Fix the base threshold $s$ and define
\[
L_c=\setsize*{\set{k\geq0:x_c>s+k\Delta_v}}.
\]
Thus, $L_c$ counts how many of the successive thresholds $s,s+\Delta_v,s+2\Delta_v,\ldots$ are smaller than $x_c$. The definition of $\Delta_v$ implies $L_i-L_j\leq1$ whenever $i\succ_vj$: a more-preferred candidate cannot remain marked for two more thresholds than a less-preferred candidate. Put $H_k=H(s+k\Delta_v)$, $H_{k,v}=H_v(s+k\Delta_v)$, $p_k=p(s+k\Delta_v)$, and $a_k=a_v(s+k\Delta_v)$, and set $a_{-1}=1$.

If $p_k<1$, then $H_{k,v}\subseteq[0,a_{k-1}]$. For $k=0$, this is immediate from $a_{-1}=1$. For $k\geq1$, a violation would place some candidate with $L_c\geq k+1$ above the candidate with $L_c\leq k-1$ that stopped the preceding scan, contradicting $L_i-L_j\leq1$. Hence, $a_{k-1}$ is an upper cutoff at threshold $s+k\Delta_v$. Applying \Cref{lem:potential-local-correction,eq:potential-internal-condition} with $a=a_k$ and $A=a_{k-1}$ bounds $Q_v(s+k\Delta_v)$ by $W(a_{k-1})-W(a_k)$.

If $p_k=1$ and $k\geq1$, every supported candidate has $L_c\geq k+1\geq2$, whereas $L_o=0$. The inequality $L_i-L_j\leq1$ forces every supported candidate below $o$, so $r_v=1$ and $Q_v(s+k\Delta_v)=\mathcal G_\phi(1)=0$. Thus, only $k=0$ can give a nonzero all-support contribution. In that case, any candidate in $H_1$ has $L_c\geq2$ and must lie below $o$, so $H_{1,v}\subseteq[0,r_v]$. Hence, $r_v$ is an upper cutoff at the next threshold. Any intervening all-support contributions vanish; when the marked probability first drops below one, its marked region still lies inside $[0,r_v]$, so the subsequent bounds from \Cref{eq:potential-internal-condition} telescope from $W(r_v)$ to zero. Therefore, \Cref{eq:potential-boundary-condition} gives
\[
\sum_{k\geq0}Q_v(s+k\Delta_v)\leq\mathcal G_\phi(r_v)+W(r_v)\leq W(1).
\]
If no threshold has marked probability one, the bounds from \Cref{eq:potential-internal-condition} telescope to $W(1)$ because $a_{-1}=1$, while $a_k=0$ for all sufficiently large $k$. Integrating over the base threshold yields
\[
\int_0^\infty Q_v(t)\diff t
=\int_0^{\Delta_v}\sum_{k\geq0}Q_v(s+k\Delta_v)\diff s
\leq W(1)\Delta_v,
\]
which proves \Cref{eq:potential-telescoping}.

It remains to average over voters and integrate over thresholds. \Cref{lem:potential-local-correction} gives $\E_v[Q_v(t)]=\E_v[a_v(t)]$ when $p(t)<1$ and $\E_v[Q_v(t)]\geq\E_v[a_v(t)]$ when $p(t)=1$. All functions are bounded step functions with finite support in $t$, so Fubini's theorem and \Cref{eq:potential-telescoping} give
\begin{equation}\label{eq:potential-excess-bound}
\loss(P,x)=\int_0^\infty\E_v[a_v(t)]\diff t\leq W(1)\E_v[\Delta_v].
\end{equation}
Because \Cref{eq:potential-excess-bound} holds for every $o$ and every nonnegative $x$ with $x_o=0$, \Cref{thm:biased-metrics} gives $\dist(P)\leq1+2W(1)$.
\end{proof}

\section{\texorpdfstring{A Rule With Metric Distortion At Most $2.1441$}{A Rule With Metric Distortion At Most 2.1441}}\label{sec:main-bound}

We apply \Cref{lem:rsl-potential} to the exact draw distribution in \Cref{fig:single-rsl-candidate}. A numerical search suggested its sparse support and a value near $2.14401485303$, but only the exact rational bound in the figure is claimed and verified over the complete continuum.

\refstepcounter{figure}\label{fig:single-rsl-candidate}
\begin{tcolorbox}[colback=gray!4,colframe=black,rounded corners,title={Figure~\thefigure: Candidate draw distribution and distortion bound}]
\small
Let $D^\star$ be supported on $\set{1,2,9,10}$ with
\[
\begin{aligned}
\Pr[D^\star=1]&=\frac{462688493532}{10^{12}},&\qquad \Pr[D^\star=2]&=\frac{194607598780}{10^{12}},\\
\Pr[D^\star=9]&=\frac{156952022902}{10^{12}},&\qquad \Pr[D^\star=10]&=\frac{185751884786}{10^{12}}.
\end{aligned}
\]
Its probability-generating function, its integral, and its stability benchmark are
\[
\begin{aligned}
\phi(u)&=\frac{462688493532u+194607598780u^2}{10^{12}}\\
&\quad+\frac{156952022902u^9+185751884786u^{10}}{10^{12}},\\
\Phi(u)&=\int_0^u\phi(s)\diff s,\qquad \mu=\Phi(1)=\frac{54251205298963}{165000000000000}.
\end{aligned}
\]
The potential parameters and candidate distortion bound are
\[
t_0=\frac{508613}{1000000},\qquad
C_0=\frac{57200742653}{100000000000},\qquad
1+2C_0=\frac{107200742653}{50000000000}=2.14401485306.
\]
\end{tcolorbox}

\mainupperbound*

\begin{proof}
Let $D^\star$, $\phi$, $\Phi$, $\mu$, $t_0$, and $C_0$ be as specified in \Cref{fig:single-rsl-candidate}. Define
\begin{equation}\label{eq:single-rsl-potential}
W(u)=
\begin{cases}
0,&0\leq u\leq t_0,\\
C_0-u-\dfrac{\mu-\phi(u)}{1-\mu},&t_0<u\leq1.
\end{cases}
\end{equation}

By \Cref{lem:rsl-potential}, it suffices to verify four claims: the endpoints of $W$, the boundary condition \Cref{eq:potential-boundary-condition}, monotonicity of $W$, and the internal condition \Cref{eq:potential-internal-condition}. All sign checks below use exact rational arithmetic. The polynomial inequalities are verified using the Bernstein subdivision method of \citet{MN13}.\footnote{They express a polynomial on the current box in the so-called Bernstein basis, whose coefficients bound its values throughout the box. If these bounds do not establish the desired sign, they bisect the box, update the Bernstein coefficients on the two subboxes, and repeat.} The accompanying repository\footnote{\url{https://github.com/nisarg89/metric-distortion}} contains the static checker \path{certificate_distortion.py}. It verifies the recorded polynomial identities, every subdivision, and the required coefficient signs without performing a search. \Cref{app:single-rsl-certificate} maps each check below to its name in the program.

\paragraph{Verification 1: Endpoints.} The displayed weights are positive and sum to one. By definition, $W(0)=0$; because $\phi(1)=1$, \Cref{eq:single-rsl-potential} gives $W(1)=C_0$.

\paragraph{Verification 2: The boundary condition.} When $u>t_0$, substituting \Cref{eq:single-rsl-potential} makes \Cref{eq:potential-boundary-condition} an equality. When $u\leq t_0$, the boundary condition is equivalent to
\begin{equation}\label{eq:single-rsl-low-boundary}
C_0-u-\frac{\mu-\phi(u)}{1-\mu}\geq0\qquad\text{for every }u\in[0,t_0].
\end{equation}
The checker verifies this polynomial inequality throughout the interval. In particular, at $u=t_0$, it shows that the right limit of $W$ is at least $W(t_0)=0$.

\paragraph{Verification 3: Monotonicity.} The first branch of $W$ is constant. Exact rational arithmetic verifies $\phi'(t_0)>1-\mu$. Because $\phi$ is convex, for $u>t_0$ we have $W'(u)=-1+\phi'(u)/(1-\mu)>0$. Together with the junction check from Verification~2, this proves that $W$ is nondecreasing.

\paragraph{Verification 4: The internal condition.} The case $p=0$ is immediate from the definition of $\mathcal E_\phi$. Suppose first that $0<p<1$ and $A\leq t_0$. Put $d=p-a$ and $K=p\mu-\Phi(p)>0$. Multiplying \Cref{eq:potential-edge} by $K$ and using that $\Phi(x)/x$ is nondecreasing gives
\begin{align*}
K\mathcal E_\phi(a,p,A)&=p\bigl(\Phi(A)-\Phi(A-d)\bigr)+p\Phi(a)-a\Phi(p)-pd\mu\\
&\leq pd\bigl(\phi(A)-\mu\bigr)\leq0.
\end{align*}
The last inequality uses the exact check $\phi(t_0)<\mu$ and the monotonicity of $\phi$. Hence $\mathcal E_\phi(a,p,A)\leq0=W(A)-W(a)$ throughout the low region.

On the remaining region, clear the positive denominator in \Cref{eq:potential-edge}. The required slack is
\begin{align}\label{eq:single-rsl-slack}
S(a,p,A)={}&\bigl(W(A)-W(a)-a\bigr)\bigl(p\mu-\Phi(p)\bigr)\notag\\
&-p\bigl(\Phi(a)+\Phi(A)-\Phi(a+A-p)-p\mu\bigr).
\end{align}
The breakpoint $t_0$ leaves three cases outside the low region: $a\leq p\leq t_0\leq A$, $a\leq t_0\leq p\leq A$, and $t_0\leq a\leq p\leq A$. After mapping each case to the unit cube and cancelling factors that are nonnegative throughout the relevant domain, the checker verifies that the resulting rational polynomial is nonnegative. Thus $S\geq0$ in all three cases. The face $A=t_0$ is covered by the analytic low-region argument; on the face $a=t_0$, using the right branch for $W(a)$ only strengthens the inequality because the actual value is $W(t_0)=0$.

All four claims of \Cref{lem:rsl-potential} now hold with $W(1)=C_0$, so every $\RSL{D^\star}$ has the candidate distortion bound in \Cref{fig:single-rsl-candidate}. On every profile, such a lottery exists by \Cref{thm:comparison-existence}; choosing any one of them defines a randomized rule with the claimed distortion.
\end{proof}

\paragraph{How the distribution was found.} The numerical discovery phase is separate from the exact proof:
\begin{itemize}
    \item The code \path{search_distortion.py} enumerated all $\sum_{i=1}^4\binom{20}{i}=6195$ subsets of $\set{1,\ldots,20}$ of size at most four and, for each support, numerically optimized its weights against a finite-grid longest-path relaxation of the potential criterion.
    \item We reoptimized the best $100$ supports on a finer grid and applied differential evolution followed by sequential least squares programming (SLSQP) to the best $20$. This produced the support $\set{1,2,9,10}$ and the numerical parameters that were then rationalized. These finite-grid calculations neither certify a continuum bound nor establish global optimality.
    \item The generic adaptive verifier \path{verify_distortion.py} constructed the fixed exact certificate described above. The theorem relies only on the rational parameters and the non-searching checker \path{certificate_distortion.py}.
\end{itemize}

\section{A Lower Bound for Every Integer Draw Distribution}\label{sec:integer-lower}

The preceding upper bound uses a carefully chosen distribution of the number of draws. We now ask how much the guarantee could improve by choosing a different positive-integer-valued distribution while retaining random-size stability. The answer is limited: no fixed draw distribution can achieve distortion at most $2.136$.

\integerlowerbound*

The proof has two conceptual ingredients. First, we construct two symmetric families of preference profiles. In the first, one candidate is inserted at one of two positions among exchangeable candidates; in the second, a distinguished set of candidates occupies one of two pairs of percentile intervals. In each case, random-size stability pins down the uniform lottery, while the offset characterization turns its position in the rankings into an exact distortion ratio. Whenever $D$ is not deterministically one, the probability-generating function $x\mapsto\E[x^D]$ is strictly convex, which makes the stable lottery unique and removes any dependence on the rule's selection convention. The idealized voter probabilities are then approximated by a finite electorate without losing the strict inequality.

Second, we reweight the probability of choosing each possible number $d$ of draws by $1/(d+1)$. After this reweighting, whether each of the three ratios exceeds $71/125$ can be expressed by the sign of the expectation of an explicit function of $d$. We exhibit positive weights for which the weighted sum of these three functions is strictly negative for every positive integer $d$. Its expectation is therefore strictly negative for every draw distribution, so at least one of the three profiles has excess-cost ratio greater than $71/125$, and hence distortion greater than $1+2\cdot71/125=267/125$. The complete constructions, the finite-profile transfer, and the exact calculation appear in \Cref{app:integer-lower-proof}.

\section{Discussion}\label{sec:discussion}

There are now two distinct gaps. Within the family of rules that always return a random-size stable lottery $\RSL{D}$ for some fixed integer draw distribution $D$, our lower and upper bounds are less than $0.009$ apart. Closing this gap may require a more refined search over draw distributions and potentials and possibly tightening our potential-based method of bounding the loss. A near-optimal draw distribution may suggest a candidate for the optimal draw distribution, which may then be provable analytically.

The second possible gap is the one between our restricted family and the broader family of all randomized voting rules. While our upper bound is less than $0.04$ above the universal (asymptotic) lower bound of \citet{CR22}, our lower bound implies that no draw distribution can completely close this gap. That said, it is possible that the bound of \citet{CR22} is not tight and an optimal randomized rule in fact lies within our family. However, if it is tight, closing the remaining gap will require seeking a rule outside of our family. 

\section*{AI Disclosure}
All the mathematical proofs were derived by OpenAI Codex (GPT-5.6-Sol at Max effort) based on research directions, literature connections, proof and search strategies, and inspirations supplied by the author. The author has verified all mathematical details, presented simplified arguments and detailed expositions, often with the aid of GPT-5.6-Sol and Claude Opus 5, and retains full responsibility for all the content and any errors. 

\section*{Acknowledgments}
This work was supported by an NSERC Discovery Grant and an NSERC-CSE Research Communities Grant. Researchers funded through the NSERC-CSE Research Communities Grants do not represent the Communications Security Establishment Canada or the Government of Canada. Any research, opinions or positions they produce as part of this initiative do not represent the official views of the Government of Canada.

\printbibliography

\clearpage
\section*{\centering Appendix}
\appendix
\addtocontents{toc}{\protect\setcounter{tocdepth}{1}}

\section{Map of the Exact Upper-Bound Certificate}\label{app:single-rsl-certificate}

This section records how the mathematical checks in the proof of \Cref{thm:intro-main-upper} appear in \path{certificate_distortion.py}. Substituting the polynomials $\phi$, $\Phi$, and $W$ from \Cref{fig:single-rsl-candidate,eq:single-rsl-potential} into the cited expressions, including the slack $S$ defined in \Cref{eq:single-rsl-slack}, and expanding gives the program's polynomials; those routine expansions are omitted.

\begin{table}[htb!]
\centering
\small
\begin{tabularx}{\textwidth}{@{}p{0.29\textwidth}p{0.27\textwidth}X@{}}
\toprule
Check in the proof & Checker routine & Polynomial in the checker\\
\midrule
The probabilities are positive and sum to one, and the distortion bound is $1+2C_0$ & \path{exact_parameter_checks} & \path{probability_mass_identity} and its reverse, \path{probability_0_positive} through \path{probability_3_positive}, and \path{distortion_identity} and its reverse\\
$\phi(t_0)<\mu$ and $\phi'(t_0)>1-\mu$ & \path{threshold_checks} & \path{activation_below_mean} and \path{threshold_derivative}\\
\Cref{eq:single-rsl-low-boundary} on $[0,t_0]$ & \path{boundary_inequality} & \path{boundary_polynomial}\\
Monotonicity of the second branch of $W$ on $[t_0,1]$ & \path{potential_monotonicity} & \path{potential_derivative}, a direct polynomial check of the conclusion also obtained from convexity in Verification~3\\
$a\leq p\leq t_0\leq A$ & \path{internal_cell_001} & \path{cell_001}\\
$a\leq t_0\leq p\leq A$ & \path{internal_cell_011} & \path{cell_011}\\
$t_0\leq a\leq p\leq A$ & \path{internal_cell_111} & \path{cell_111}\\
\bottomrule
\end{tabularx}
\end{table}

For completeness, the three cell polynomials use variables $(x,y,z)\in[0,1]^3$. The polynomial \path{cell_001} is the polynomial continuation of $S(a,p,A)/p$ after substituting $a=t_0xy$, $p=t_0y$, and $A=t_0+(1-t_0)z$. The polynomial \path{cell_011} is the polynomial continuation of $S(a,p,A)/(1-p)$ after substituting $a=t_0x$, $p=t_0+(1-t_0)y$, and $A=p+(1-p)z$. Finally, \path{cell_111} is the polynomial continuation of $S(a,p,A)/((1-a)(1-p))$ after substituting $a=t_0+(1-t_0)x$, $p=a+(1-a)y$, and $A=p+(1-p)z$. The cancelled factors are positive in the interiors of the corresponding cells, and the polynomial continuations cover their boundary faces.

\section{Further Properties of Random-Size Stable Lotteries}\label{app:unification}

This appendix records further consequences of random-size stability. It relates the definition to deterministic challenger committees, comparisons using other numbers of samples, and lotteries with small support. These consequences do not subsume stable lotteries over committees \citep{CJMW20,JMW20}: here every committee consists of independent draws from one lottery and a voter compares committees by their favorite members, whereas that literature permits correlated committee lotteries and more general committee preferences.

\subsection{Why separately chosen stable lotteries cannot simply be mixed}\label{app:mixing-counterexample}

Suppose masses $3/5$ and $2/5$ of the voters rank $a\succ b$ and $b\succ a$, respectively. The lottery $P_1=\delta_a$ is stable for one draw, while $P_2=(4/5)\delta_a+(1/5)\delta_b$ is stable for two draws because
\begin{align*}
\Pr\nolimits_v[a\succ_v P_2^2]
&=\frac35\int_0^1\left(\frac15+\frac45s\right)^2\diff s
+\frac25\int_0^1\left(\frac45s\right)^2\diff s=\frac13,\\
\Pr\nolimits_v[b\succ_v P_2^2]
&=\frac35\int_0^1\left(\frac15s\right)^2\diff s
+\frac25\int_0^1\left(\frac45+\frac15s\right)^2\diff s=\frac13.
\end{align*}
If $D$ is uniform on $\set{1,2}$, their average $P=(P_1+P_2)/2$ satisfies
\[
\Pr\nolimits_v[a\succ_v P]=\frac{51}{100},
\qquad
\Pr\nolimits_v[a\succ_v P^2]=\frac{33}{100},
\]
and hence $\Pr\nolimits_v[a\succ_v P^D]=21/50>5/12=\E[1/(D+1)]$. The cross terms created by drawing repeatedly from $P$ are not controlled by the separate guarantees for $P_1$ and $P_2$. Selecting $P_D$ after observing $D$ would satisfy a different averaged statement, not the requirement that one base lottery be chosen before $D$.

\subsection{Computing a random-size stable lottery}

\begin{proposition}\label{prop:comparison-computation}
Suppose $D$ has finite support of size $s$, whose largest value is $d_{\max}$, and let $m=|C|\geq2$. For every positive integer $N$, one can compute a lottery $\overline P$ satisfying
\[
\max_{c\in C}\Pr\nolimits_v[c\succ_v\overline P^D]
\leq\E\left[\frac1{D+1}\right]+\sqrt{\frac{\log m}{2N}}.
\]
In the real-arithmetic model, the algorithm uses $O(Nnms\log(d_{\max}+1))$ arithmetic operations and exponential evaluations. Thus additive error at most $\varepsilon$ is obtained in $O(nms\log(d_{\max}+1)\log(m)/\varepsilon^2)$ time.
\end{proposition}

\begin{proof}
Start with the uniform lottery $P_1$ and set $\eta=\sqrt{8\log(m)/N}$. At iteration $t$, put $g_{t,c}=\Pr\nolimits_v[c\succ_v P_t^D]$ and update
\[
P_{t+1}(c)=
\frac{P_t(c)\exp(\eta g_{t,c})}
{\sum_{b\in C}P_t(b)\exp(\eta g_{t,b})}.
\]
The standard exponential-weights calculation for numbers in $[0,1]$ gives
\[
\max_{c\in C}\frac1N\sum_{t=1}^N g_{t,c}
\leq
\frac1N\sum_{t=1}^N\sum_{c\in C}P_t(c)g_{t,c}
+\frac{\log m}{\eta N}+\frac\eta8.
\]
The exchangeability identity proved inside \Cref{thm:comparison-existence} makes the inner sum $\E[1/(D+1)]$ at every iteration. With the chosen $\eta$, the last two terms sum to $\sqrt{\log(m)/(2N)}$. Finally, let $\overline P=N^{-1}\sum_{t=1}^NP_t$. The convexity established in the same proof gives $\Pr\nolimits_v[c\succ_v\overline P^D]\leq N^{-1}\sum_t\Pr\nolimits_v[c\succ_v P_t^D]$ for every $c$.

It remains to account for the running time. For a voter $v$ and candidate $c$, let $b_{v,c}$ be the $P_t$-mass strictly below $c$. Because every $P_t(c)$ is positive, \Cref{lem:compare-c-PD} gives
\[
g_{t,c}=\frac1n\sum_{v\in V}\sum_{d\in\operatorname{supp}(D)}\Pr[D=d]\,
\frac{(b_{v,c}+P_t(c))^{d+1}-b_{v,c}^{d+1}}{(d+1)P_t(c)}.
\]
For each voter, all values $b_{v,c}$ are obtained in one scan of her ranking. Repeated squaring then evaluates the displayed powers in $O(nms\log(d_{\max}+1))$ arithmetic operations per iteration; the multiplicative-weights update costs another $O(m)$ operations and exponential evaluations. This proves the claimed total.
\end{proof}

For an infinite-support distribution, the same iteration guarantee holds provided that the required expectations can be evaluated to sufficient precision, but the running time then depends on the cost of those evaluations.

For a deterministic nonempty committee $A\subseteq C$, the event $A\succ_v(\RSL{D})^D$ means that voter $v$'s favorite labeled member of $A$ beats her favorite of the $D$ independent draws.

\begin{proposition}\label{prop:set-attack}
Every nonempty committee $A\subseteq C$ satisfies
\[
\Pr\nolimits_v[A\succ_v(\RSL{D})^D]
\leq\min\set{1,|A|\E\left[\frac1{D+1}\right]}.
\]
Moreover,
\[
\min_{P\in\Delta(C)}\max_{\varnothing\neq A\subseteq C}
\frac{\Pr\nolimits_v[A\succ_v P^D]}{|A|}
=\E\left[\frac1{D+1}\right].
\]
\end{proposition}

\begin{proof}
The event that $A$ defeats $(\RSL{D})^D$ is the union, over $c\in A$, of the event that the labeled copy of $c$ defeats $(\RSL{D})^D$. The union bound and \Cref{eq:comparison-system} give the first inequality. For any lottery $P$, the exchangeability identity in the proof of \Cref{thm:comparison-existence} gives some candidate $c$ with $\Pr\nolimits_v[c\succ_v P^D]\geq\E[1/(D+1)]$, which proves the matching lower bound for the normalized maximum.
\end{proof}

Fix a committee $A$, expose the voter and the labeled copies in $A$, and let $T\in[0,1]$ be the probability that one draw from $\RSL{D}$ lies below the voter's favorite labeled member of $A$. Then $\Pr\nolimits_v[A\succ_v(\RSL{D})^r]=\E[T^r]$ and, by the conditioning calculation in \Cref{lem:compare-c-PD}, $\Pr\nolimits_v[A\succ_v(\RSL{D})^D]=\E[\phi(T)]$.

Because $D$ is positive-integer-valued, $\phi$ is continuous and strictly increasing from $0$ to $1$ on $[0,1]$, so its inverse $\phi^{-1}:[0,1]\to[0,1]$ is well-defined.

\begin{lemma}\label{lem:higher-moments}
Let $s_A=\min\set{1,|A|\E[1/(D+1)]}$. For an integer $r\geq1$, define $h_{r,\phi}(u)=(\phi^{-1}(u))^r$, and let $\cav h_{r,\phi}$ be the smallest concave function on $[0,1]$ lying above $h_{r,\phi}$. Then
\[
\Pr\nolimits_v[A\succ_v(\RSL{D})^r]\leq(\cav h_{r,\phi})(s_A).
\]
If $D=d$ deterministically, this bound is $s_A^{r/d}$ for $r\leq d$ and $s_A$ for $r\geq d$.
\end{lemma}

\begin{proof}
By \Cref{prop:set-attack}, $\E[\phi(T)]\leq s_A$. The function $\cav h_{r,\phi}$ is nondecreasing, since it is at most the constant concave upper bound $1$ and equals $1$ at the right endpoint. Because $T^r=h_{r,\phi}(\phi(T))$, the definition of the concave upper bound and Jensen's inequality give
\[
\E[T^r]\leq\E[(\cav h_{r,\phi})(\phi(T))]
\leq(\cav h_{r,\phi})(\E[\phi(T)])
\leq(\cav h_{r,\phi})(s_A).
\]
For $\phi(t)=t^d$, the function $u^{r/d}$ is concave when $r\leq d$; when $r\geq d$, its smallest concave upper bound on $[0,1]$ is the line $u$.
\end{proof}

Let $B$ be a lottery, let $R$ be an independent positive integer, and write $\psi(t)=\E[t^R]$.

\begin{theorem}\label{thm:two-generating-functions}
Every random-size stable lottery satisfies
\[
\Pr\nolimits_v[B^R\succ_v(\RSL{D})^D]
\leq1-\psi\left(1-\E\left[\frac1{D+1}\right]\right).
\]
For deterministic $R=r$, the bound becomes $1-(1-\E[1/(D+1)])^r$.
\end{theorem}

Finally, random-size stable lotteries can be approximated in support, although the resulting additive error is too coarse for the main distortion proof. Let $\bar d=\E[D]=\phi'(1)$.

\begin{proposition}\label{prop:sampling-rsl}
For every integer $s\geq1$, there is a lottery $\widehat P$ supported on at most $s$ candidates such that
\[
\Pr\nolimits_v[c\succ_v\widehat P^D]
\leq\E\left[\frac1{D+1}\right]+\bar d\sqrt{\frac\pi{2s}}
\qquad\text{for every }c\in C.
\]
\end{proposition}

We now prove the last two statements and then record an extension to signed coefficients. The signed extension is mathematically useful but does not describe a random value of $D$ unless the coefficients are nonnegative and sum to one.

\subsection{A calculation with cumulative distributions}

Draw a uniformly random voter $v$ and order labeled candidates from worst to best according to $v$. For lotteries $A$ and $P$, let $F_{A,v}$ and $F_{P,v}$ be their cumulative distribution functions in this order. Draw a labeled outcome $\omega$ from $A$ and put
\[
U=F_{P,v}(\omega),\qquad Y=F_{A,v}(\omega).
\]
Conditional on $v$, $Y$ is uniform on $[0,1]$. For every integer $d\geq1$,
\begin{equation}\label{eq:appendix-cdf-powers}
\Pr\nolimits_v[A\succ_v P^d]=\E[U^d],
\qquad
\Pr\nolimits_v[A^d\succ_v P]=d\E[UY^{d-1}].
\end{equation}
The first identity conditions on the draw from $A$. For the second, the best of $d$ independent $A$-draws has density $dF_{A,v}^{d-1}\diff F_{A,v}$. Summing \Cref{eq:appendix-cdf-powers} over the distribution of $D$, as in the conditioning calculation of \Cref{lem:compare-c-PD}, gives the formulas used in \Cref{lem:reverse-identity}:
\[
\Pr\nolimits_v[A\succ_v P^D]=\E[\phi(U)],
\qquad
\Pr\nolimits_v[A^D\succ_v P]=\E[U\phi'(Y)].
\]

\begin{lemma}\label{lem:reverse-identity}
For finitely supported $D$, define
\[
\mathcal D_\phi(u,v)=\phi(u)-\phi(v)-\phi'(v)(u-v).
\]
Then
\begin{equation*}
\Pr\nolimits_v[A^D\succ_v P]
=\Pr\nolimits_v[A\succ_v P^D]+1-2\E\left[\frac1{D+1}\right]-\E[\mathcal D_\phi(U,Y)].
\end{equation*}
Consequently, $\Pr\nolimits_v[A\succ_v P^D]\leq\E[1/(D+1)]$ implies $\Pr\nolimits_v[A^D\succ_v P]\leq1-\E[1/(D+1)]$. The probability inequality extends to arbitrary $D$ by truncation.
\end{lemma}

\begin{proof}
Convexity gives $\mathcal D_\phi\geq0$, and rearranging its definition gives
\[
U\phi'(Y)=\phi(U)-\mathcal D_\phi(U,Y)+Y\phi'(Y)-\phi(Y).
\]
Because $Y$ is uniform conditional on every voter,
\[
\E[Y\phi'(Y)-\phi(Y)]
=\int_0^1(t\phi'(t)-\phi(t))\diff t
=1-2\int_0^1\phi(t)\diff t
=1-2\E\left[\frac1{D+1}\right].
\]
Substitution into the two preceding cumulative-distribution formulas proves the identity, and discarding the nonnegative term proves the consequence. Equality holds when $A=P$ by exchangeability.
\end{proof}

\subsection{Several draws from another lottery}

For this proof, abbreviate $\mu=\E[1/(D+1)]$. Draw a labeled $\omega$ from $B$. The best of $R$ draws from $B$ has density $\psi'(F_{B,v}(\omega))$ with respect to a draw from $B$. Conditional on this best draw, the calculation in \Cref{lem:compare-c-PD} shows that all $D$ draws from $\RSL{D}$ lie below it with probability $\phi(F_{\RSL{D},v}(\omega))$. With $U=F_{\RSL{D},v}(\omega)$ and $Y=F_{B,v}(\omega)$, therefore,
\begin{equation}\label{eq:two-pgf-integral}
\Pr\nolimits_v[B^R\succ_v(\RSL{D})^D]=\E[\phi(U)\psi'(Y)].
\end{equation}

By linearity of \Cref{eq:comparison-system} in the candidate drawn from $B$,
\[
\E[\phi(U)]=\Pr\nolimits_v[B\succ_v(\RSL{D})^D]\leq \mu.
\]
Also $0\leq\phi(U)\leq1$, and $Y$ is uniform on $[0,1]$. The function $\psi'$ is nondecreasing. We claim that for every random pair $(Z,Y)$ with $0\leq Z\leq1$, $\E Z\leq \mu$, and uniform marginal $Y$,
\begin{equation}\label{eq:bathtub-bound}
\E[Z\psi'(Y)]\leq\int_{1-\mu}^1\psi'(t)\diff t.
\end{equation}
To prove this directly, put $\vartheta=\psi'(1-\mu)$. Pointwise,
\[
Z(\psi'(Y)-\vartheta)\leq(\psi'(Y)-\vartheta)_+.
\]
Taking expectations and using $\E Z\leq \mu$ gives
\[
\E[Z\psi'(Y)]
\leq\vartheta \mu+\int_0^1(\psi'(t)-\vartheta)_+\diff t
=\int_{1-\mu}^1\psi'(t)\diff t,
\]
where a flat portion at the threshold does not affect the equality. Applying \Cref{eq:bathtub-bound} to $Z=\phi(U)$ in \Cref{eq:two-pgf-integral} yields
\[
\Pr\nolimits_v[B^R\succ_v(\RSL{D})^D]
\leq\int_{1-\mu}^1\psi'(t)\diff t
=1-\psi(1-\mu).
\]
For deterministic $R=r$, substitute $\psi(t)=t^r$.

\subsection{Approximating by a lottery with small support}

Draw $s$ candidates independently from $\RSL{D}$, and let $\widehat P$ be their empirical distribution. Fix a voter. Place the candidates in this voter's order and linearly interpolate across the interval corresponding to each candidate; this interpolation is exactly the proportional label convention. Let $F$ and $\widehat F$ be the resulting cumulative distribution functions for $\RSL{D}$ and $\widehat P$. The sharp Dvoretzky--Kiefer--Wolfowitz inequality \citep{Mas90} gives
\[
\Pr\left[\sup_t|\widehat F(t)-F(t)|>\varepsilon\right]
\leq2e^{-2s\varepsilon^2}.
\]
Integrating the tail bound gives
\begin{equation}\label{eq:expected-DKW}
\E\left[\sup_t|\widehat F(t)-F(t)|\right]
\leq\int_0^\infty2e^{-2s\varepsilon^2}\diff\varepsilon
=\sqrt{\frac{\pi}{2s}}.
\end{equation}
The same sampled candidates are used for every voter; \Cref{eq:expected-DKW} is applied separately and then averaged over voters.

Because $\phi'(t)\leq\phi'(1)=\bar d$ on $[0,1]$, the function $\phi$ is $\bar d$-Lipschitz. By \Cref{lem:compare-c-PD}, for any pure candidate $c$, the probability that $c$ defeats the empirical committee differs from the corresponding probability under $\RSL{D}$ by at most $\bar d\sup_t|\widehat F(t)-F(t)|$ for that voter. Therefore
\[
\E_{\widehat P}\left[
\max_{c\in C}\bigl(\Pr\nolimits_v[c\succ_v\widehat P^D]-\Pr\nolimits_v[c\succ_v(\RSL{D})^D]\bigr)
\right]
\leq\bar d\sqrt{\frac{\pi}{2s}}.
\]
The maximum causes no factor depending on $|C|$: for each voter, the same uniform error in the cumulative distributions controls every candidate position. Since $\Pr\nolimits_v[c\succ_v(\RSL{D})^D]\leq\E[1/(D+1)]$ for every $c$, some realization of the sample satisfies
\[
\Pr\nolimits_v[c\succ_v\widehat P^D]\leq\E\left[\frac1{D+1}\right]+\bar d\sqrt{\frac{\pi}{2s}}
\]
simultaneously for all candidates. Its empirical distribution has support at most $s$.

\subsection{An extension to signed coefficients}

\begin{theorem}\label{thm:signed-comparison}
Let $(a_d)_{d\geq1}$ be an absolutely summable sequence of real numbers, define
\[
f(t)=\sum_{d\geq1}a_dt^d,
\]
and let $p_f=\int_0^1f(t)\diff t$. Every finite profile admits a lottery $P$ such that
\[
\sum_{d\geq1}a_d\Pr\nolimits_v[c\succ_v P^d]\leq\int_0^1 f(t)\diff t
\]
for every candidate $c$. This right-hand side is best possible: for every lottery $P$, the maximum of the left-hand side over candidates is at least $\int_0^1f(t)\diff t$.
\end{theorem}

\begin{proof}
Within this proof, write $G_c(P)$ for the left-hand side. Since $p_f=\sum_{d\geq1}a_d/(d+1)$, absolute summability gives uniform convergence of the defining series on $[0,1]$, hence continuity of every $G_c$. The exchangeability argument in the proof of \Cref{thm:comparison-existence} is linear in the coefficients and gives
\[
\sum_cP(c)G_c(P)=p_f.
\]
For a lottery $P$, let $M(P)=\operatorname*{arg\,max}_{c\in C}G_c(P)$, and let $B(P)$ be the set of lotteries supported on $M(P)$. Each $B(P)$ is nonempty, convex, and compact. Continuity of the functions $G_c$ implies that the correspondence $B$ has a closed graph. Kakutani's fixed-point theorem therefore gives a lottery $P$ with $P\in B(P)$. Every candidate in the support of this lottery maximizes $G_c(P)$, so the displayed identity gives
\[
\max_cG_c(P)=\sum_cP(c)G_c(P)=p_f.
\]
For every lottery, the maximum of the numbers $G_c(P)$ is at least their $P$-weighted average $p_f$, which proves the converse.
\end{proof}

Signed coefficients do not define probabilities and need not make the weighted comparison function convex in $P$, so the minimax proof of \Cref{thm:comparison-existence} does not apply. They also do not imply the union bound or the consequence of \Cref{lem:reverse-identity} used in \Cref{prop:set-attack}. We use nonnegative coefficients summing to one everywhere in the distortion results.

\section{Proof of the Lower Bound}\label{app:integer-lower-proof}

This section proves \Cref{thm:intro-integer-lower}. An \emph{idealized profile} assigns arbitrary real-valued probabilities to a finite collection of rankings. We first show that a nondegenerate draw distribution admits at most one random-size stable lottery on every profile. We then give two idealized profile constructions whose distortion can be read exactly from the offset characterization, transfer their strict inequalities to ordinary finite electorates, and finish with an exact rational calculation showing that one of three profiles is hard, no matter how the number of draws is distributed.

Throughout this section, fix a random variable $D$ supported on the positive integers. Recall its probability-generating function $\phi$ from \Cref{def:probability-generating-function}, and its integral $\Phi$ and the benchmark $\mu=\Phi(1)=\E[1/(D+1)]$ from \Cref{def:potential-draw-functions}.

For a lottery $P$ and a candidate $c$, let $b_{v,c}(P)$ be the $P$-mass that voter $v$ ranks strictly below $c$. If $U$ is uniform on $[0,1]$, then \Cref{lem:compare-c-PD} gives the comparison score of $c$ against $P$ as
\[
L_c(P)=\E_{v,U}\bigl[\phi(b_{v,c}(P)+UP(c))\bigr]
=\Pr\nolimits_v[c\succ_v P^D].
\]
Thus $P$ is random-size stable if and only if $L_c(P)\leq\mu$ for every candidate $c$. By the self-play identity in the proof of \Cref{thm:comparison-existence}, $\sum_cP(c)L_c(P)=\mu$ for every $P$, and every candidate in the support of a stable lottery has score exactly $\mu$.

\subsection{Uniqueness and two hard-profile constructions}

When $D$ is not deterministically one, the function $\phi$ is strictly convex. This removes the otherwise problematic freedom to select among several stable lotteries.

\begin{lemma}[Uniqueness]\label{lem:integer-lower-unique}
If $\Pr[D\geq2]>0$, every probability distribution over the rankings of a finite candidate set, and in particular every finite profile, admits at most one random-size stable lottery for $D$.
\end{lemma}

\begin{proof}
Suppose $P$ and $Q$ are stable and put $R=(P+Q)/2$. For every candidate $c$, both $b_{v,c}(P)$ and $P(c)$ are affine in $P$, so convexity of $\phi$ gives $L_c(R)\leq(L_c(P)+L_c(Q))/2\leq\mu$. Hence $R$ is stable. If $c\in\supp(R)$, support tightness gives $L_c(R)=\mu$, so both inequalities are equalities. Strict convexity therefore implies
\[
b_{v,c}(P)+UP(c)=b_{v,c}(Q)+UQ(c)
\]
almost surely over $v$ and $U$. Since $U$ has a continuous distribution, the slopes agree and $P(c)=Q(c)$. This holds for every candidate in $\supp(P)\cup\supp(Q)=\supp(R)$, while both lotteries vanish elsewhere. Thus $P=Q$.
\end{proof}

The two constructions are most transparent through the following ratios. Their denominators are positive: for $p\in(0,1)$, strict monotonicity of $\phi$ gives $\Phi(p)/p=\int_0^1\phi(pt)\diff t<\int_0^1\phi(t)\diff t=\mu$.

\begin{definition}[Ratios realized by the two constructions]\label{def:integer-lower-ratios}
For $r\in[0,1]$, define
\[
R_\phi(r)=r+\frac{\mu-\phi(r)}{1-\mu}.
\]
For $0\leq a\leq p<1$, define
\[
S_\phi(a,p)=a+
\frac{p\bigl(\Phi(a)+\mu-\Phi(1+a-p)-p\mu\bigr)}{p\mu-\Phi(p)}.
\]
\end{definition}

\begin{lemma}[Boundary construction]\label{lem:integer-lower-boundary}
Assume $\Pr[D\geq2]>0$. Let $r\in(0,1)$ be rational, let $N$ be a positive integer such that $rN$ is integral, and suppose $\phi(r)<\mu$. There is an idealized distribution over rankings of $N+1$ candidates whose unique random-size stable lottery $P$ has offset value $1+2R_\phi(r)$.
\end{lemma}

\begin{proof}
Let $P$ be uniform on $N$ exchangeable candidates and assign probability zero to one additional candidate $o$. Put $z=(\mu-\phi(r))/(1-\phi(r))\in(0,1)$. Choose a uniform random order of the supported candidates. With probability $z$, place $o$ above all of them; with probability $1-z$, place $o$ so that exactly $rN$ supported candidates are below it.

Candidate $o$ has comparison score $z+(1-z)\phi(r)=\mu$. Symmetry gives the same score to every supported candidate, and their $P$-weighted average is the self-play value $\mu$. Hence all candidates have score at most $\mu$, so $P$ is stable; it is unique by \Cref{lem:integer-lower-unique}.

Give $o$ offset zero and assign offset one to every supported candidate. When $o$ is ranked first, the excess-cost expression in \Cref{thm:biased-metrics} is one and $\Delta_v=0$; in the other rankings, they are $r$ and one, respectively. The corresponding offset value is
\[
1+2\frac{z+(1-z)r}{1-z}
=1+2\left(r+\frac{\mu-\phi(r)}{1-\mu}\right)
=1+2R_\phi(r).
\]
\end{proof}

\begin{lemma}[Internal construction]\label{lem:integer-lower-internal}
Assume $\Pr[D\geq2]>0$. Let $0\leq a<p<1$ be rational, let $N$ be a positive integer such that $aN$ and $pN$ are integral, and suppose $\Phi(a)+\mu-\Phi(1+a-p)>p\mu$. There is an idealized distribution over rankings of $N$ candidates whose unique random-size stable lottery $P$ has offset value $1+2S_\phi(a,p)$.
\end{lemma}

\begin{proof}
Partition the candidates into a set $H$ of size $pN$ and a set $L$ of size $(1-p)N$, and let $P$ be uniform on all $N$ candidates. Put
\[
K=p\mu-\Phi(p),\qquad M=\Phi(a)+\mu-\Phi(1+a-p),
\qquad z=\frac{K}{M-\Phi(p)}.
\]
The hypothesis gives $M>p\mu>\Phi(p)$, so $z\in(0,1)$. With probability $z$, place $aN$ members of $H$ at the bottom of the ranking, all members of $L$ next, and the remaining members of $H$ at the top. With probability $1-z$, place all of $H$ at the bottom and all of $L$ at the top. In both cases, average uniformly over all permutations within each displayed group.

For any ranking, multiplying each candidate's score by its probability and summing over a set of candidates integrates $\phi$ over the percentile intervals occupied by that set. The integral over the positions occupied by $H$ is $M$ in the first ranking type and $\Phi(p)$ in the second. Its expectation is $zM+(1-z)\Phi(p)=p\mu$. Symmetry therefore gives score $\mu$ to every member of $H$. The self-play identity then gives total weighted score $(1-p)\mu$ to $L$, so symmetry gives score $\mu$ to every member of $L$. Thus $P$ is stable, and it is unique by \Cref{lem:integer-lower-unique}.

Give candidates in $H$ offset one and candidates in $L$ offset zero, and designate any member of $L$ as $o$. The first ranking type has excess-cost expression $a$ and $\Delta_v=1$, while the second has excess-cost expression $p$ and $\Delta_v=0$. Hence the corresponding offset value is
\[
1+2\frac{za+(1-z)p}{z}
=1+2\left(a+\frac{p(M-p\mu)}{K}\right)
=1+2S_\phi(a,p).
\]
\end{proof}

The preceding constructions use real-valued probabilities on a finite collection of rankings. The next lemma transfers a strict inequality to an ordinary finite electorate and makes the selection quantifier explicit.

\begin{lemma}[Finite-profile transfer]\label{lem:integer-lower-finite}
Suppose one of the idealized profiles in \Cref{lem:integer-lower-boundary,lem:integer-lower-internal} has offset value strictly greater than $1+2c$. Then an ordinary finite profile on the same candidates has the property that every random-size stable lottery for $D$ has distortion strictly greater than $1+2c$.
\end{lemma}

\begin{proof}
Approximate the finitely many ranking probabilities by rational numbers and realize each rational distribution by a finite electorate. If the conclusion failed along every sufficiently accurate approximation, choose from each such profile a stable lottery whose offset value is at most $1+2c$. The candidate simplex is compact, so a subsequence of these lotteries converges. Candidate scores are continuous jointly in the ranking frequencies and the lottery, and therefore the limit is stable for the idealized profile. By \Cref{lem:integer-lower-unique}, it must be the uniform lottery displayed in the construction. The offset value is also continuous, and its denominator is positive at the idealized profile. The strict inequality must therefore hold throughout a sufficiently small neighborhood, a contradiction. Finally, \Cref{thm:biased-metrics} turns this value into a lower bound on distortion.
\end{proof}

\begin{lemma}[Candidate count]\label{lem:integer-lower-candidate-count}
For $r=71/125$ and for $(a,p)\in\set{(1/100,37/50),(1/50,3/4)}$, the corresponding finite witness obtained from \Cref{lem:integer-lower-finite} can be chosen with $5001$ candidates.
\end{lemma}

\begin{proof}
For every displayed parameter, $N=5000$ is valid. The boundary construction then has $5001$ candidates. The internal construction has $5000$ candidates; add one candidate of probability zero at the bottom of every ranking and give it offset one. Its comparison score is $\phi(0)=0$, and it changes neither stability nor the distortion calculation. Thus the finite transfer preserves a candidate set of size $5001$.
\end{proof}

The case when $D=1$ deterministically does not require the preceding uniqueness argument and admits a simple exact witness.

\begin{lemma}\label{lem:integer-lower-degree-one}
Every rule that returns a random-size stable lottery for deterministic $D=1$ has distortion greater than $267/125$ on a profile with $5001$ candidates.
\end{lemma}

\begin{proof}
Place candidates $a$ and $b$ at $0$ and $1$ on the real line, respectively, and place the other $4999$ candidates at $3$. Of $1000$ voters, place $501$ at $499/1000$ and $499$ at $1$, breaking the metric ties among the additional candidates arbitrarily. Candidate $a$ is the strict majority winner against $b$ and is unanimously preferred to every additional candidate. For any lottery $P\ne\delta_a$, the comparison score of $a$ against one draw from $P$ is
\[
\frac12P(a)+\frac{501}{1000}P(b)+P(C\setminus\set{a,b})>\frac12,
\]
whereas $\delta_a$ is stable. Hence the unique stable lottery for one draw, equivalently the unique maximal lottery, selects $a$. Candidate $b$ is socially optimal, while
\[
\SC(a)=\frac{748999}{10^6},\qquad
\SC(b)=\frac{251001}{10^6},\qquad
\frac{\SC(a)}{\SC(b)}=\frac{748999}{251001}>\frac{267}{125}.
\]
\end{proof}

\subsection{Why at least one of the three profiles exceeds the target}

It remains to prove that, for every draw distribution, at least one of the three instances of the preceding constructions yields a ratio greater than $71/125$. The following proposition is the only computer-assisted step in the lower bound.

\begin{proposition}[At least one of the three ratios exceeds the target]\label{prop:integer-lower-alternatives}
For every positive-integer-valued random variable $D$, at least one of the following inequalities holds:
\begin{align*}
R_\phi\left(\frac{71}{125}\right)&>\frac{71}{125},\\
S_\phi\left(\frac1{100},\frac{37}{50}\right)&>\frac{71}{125},\\
S_\phi\left(\frac1{50},\frac34\right)&>\frac{71}{125}.
\end{align*}
\end{proposition}

\begin{proof}
Put $c=71/125$ and let $\nu_d=\Pr[D=d]$. Define another probability distribution on the positive integers by
\[
\widehat\nu(d)=\frac{\nu_d}{\mu(d+1)}.
\]
Indeed, its masses sum to one, and $\E_{d\sim\widehat\nu}[d]=(1-\mu)/\mu$ is finite and positive. For every positive integer $d$, define
\[
b_d=(d+1)c^d-1
\]
and, for $0\leq a\leq p<1$, define
\[
i_d(a,p)=p\left((c-a)(1-p^d)-a^{d+1}-1+(1+a-p)^{d+1}+p\right).
\]
Substituting $\phi(x)=\sum_d\nu_dx^d$ and $\Phi(x)=\sum_d\nu_dx^{d+1}/(d+1)$ gives the exact identities
\begin{align}
\E_{d\sim\widehat\nu}[b_d]
&=\E_{d\sim\widehat\nu}[d]\bigl(c-R_\phi(c)\bigr),\label{eq:integer-lower-boundary-identity}\\
\E_{d\sim\widehat\nu}[i_d(a,p)]
&=\E_{d\sim\widehat\nu}[p(1-p^d)]\bigl(c-S_\phi(a,p)\bigr).\label{eq:integer-lower-internal-identity}
\end{align}
Both multiplicative expectations on the right are strictly positive.

Consider the positive weights
\[
(w_0,w_1,w_2)=
\left(\frac{45697}{250000},\frac{24337}{500000},\frac{384269}{500000}\right),
\]
which sum to one, and put $(a_1,p_1)=(1/100,37/50)$ and $(a_2,p_2)=(1/50,3/4)$. Exact rational arithmetic proves
\begin{equation}\label{eq:integer-lower-negative-combination}
F_d:=w_0b_d+w_1i_d(a_1,p_1)+w_2i_d(a_2,p_2)<0
\qquad\text{for every integer }d\geq1.
\end{equation}
For completeness, we describe the finite verification and the infinite tail. Direct rational evaluation proves \Cref{eq:integer-lower-negative-combination} for $1\leq d\leq17$. Both parameter pairs satisfy $1+a_j-p_j=27/100$. Expanding the left-hand side gives
\[
F_d=L+w_0(d+1)c^d+(w_1p_1+w_2p_2)\left(\frac{27}{100}\right)^{d+1}
-\sum_{j=1}^2w_j\left((c-a_j)p_j^{d+1}+p_ja_j^{d+1}\right),
\]
where $L=-w_0+\sum_{j=1}^2w_jp_j(c-a_j+p_j-1)$. After discarding the nonpositive final sum, the two remaining positive exponential terms decrease for $d\geq17$; for the first one, the ratio of successive terms is at most $(71/125)(19/18)<1$. Exact rational evaluation at $d=17$ makes the resulting upper bound smaller than $-1/15000$, proving \Cref{eq:integer-lower-negative-combination} for the entire tail.

The accompanying verifier \path{certificate_lower_bound.py} reconstructs the three functions from their definitions, performs every finite rational comparison, checks the tail expansion, and independently checks a second choice of positive weights that proves the same conclusion. The computation covers every possible positive number of draws, not merely a sampled or bounded set of values.

Because the three functions are bounded, averaging \Cref{eq:integer-lower-negative-combination} under $\widehat\nu$ remains strictly negative even when $D$ has infinite support. Consequently, at least one of the three expectations on the left of \Cref{eq:integer-lower-boundary-identity,eq:integer-lower-internal-identity} is negative. The positive factors on the right then imply that at least one of $R_\phi(c)$, $S_\phi(a_1,p_1)$, and $S_\phi(a_2,p_2)$ is strictly greater than $c$, as claimed.
\end{proof}

\begin{proof}[Proof of \Cref{thm:intro-integer-lower}]
Fix the draw distribution of $D$ before the profile is observed. If $D=1$ deterministically, \Cref{lem:integer-lower-degree-one} gives the result. Otherwise, \Cref{prop:integer-lower-alternatives} makes one of the three displayed ratios strictly greater than $71/125$. In the boundary case, this inequality implies $\phi(71/125)<\mu$, so \Cref{lem:integer-lower-boundary} applies. In either internal case, the first argument is smaller than $71/125$ and $p\mu-\Phi(p)>0$, so the inequality implies $\Phi(a)+\mu-\Phi(1+a-p)>p\mu$ and \Cref{lem:integer-lower-internal} applies. The resulting idealized profile has a unique stable lottery and offset value strictly greater than $1+2\cdot71/125=267/125$. By \Cref{lem:integer-lower-finite,lem:integer-lower-candidate-count}, an ordinary finite profile with $5001$ candidates has the same strict property for every random-size stable lottery. This proves the strict lower bound for each fixed draw distribution. Taking the infimum over all draw distributions turns these pointwise strict inequalities into the weak class lower bound $267/125$.
\end{proof}

\end{document}